\documentclass{vldb}
\usepackage{graphicx}
\usepackage{balance}  

\usepackage{xspace,textcomp}
\usepackage{tikz}
\usepackage{flushend}
\usepackage[]{algorithm2e}
\usepackage{subcaption}
\usetikzlibrary{circuits.logic.US, positioning} 
\usetikzlibrary{shapes.geometric, arrows}

\vldbTitle{Shrinkwrap: Differentially-Private Query Processing in Private Data Federations}
\vldbAuthors{Johes Bater, Xi He, William Ehrich, Ashwin Machanavajjhala, Jennie Rogers}
\vldbDOI{https://doi.org/TBD}
\vldbVolume{12}
\vldbNumber{xxx}
\vldbYear{2019}

\newtheorem{example}{Example}
\newtheorem{problem}{Problem}
\newtheorem{definition}{Definition}
\newtheorem{theorem}{Theorem}

\newcommand{\sysname}{Shrinkwrap\xspace}
\newcommand{\federation}{private data federation\xspace}
\newcommand{\hb}{query coordinator\xspace}
\newcommand{\mpc}{secure computation\xspace}
\newcommand{\epsilonper}{\epsilon_{1\rightarrow l}}
\newcommand{\deltaper}{\delta_{1\rightarrow l}}
\newcommand{\carquery}{{\bf c}}
\newcommand{\tlap}{\text{TLap}}

\newcommand{\stitle}[1]{\smallskip \noindent{\bf #1}}
\newcommand{\eat}[1]{}
\newcommand{\change}[1]{{\color{black} #1}}
\newcommand{\am}[1]{\emph{{\color{red} [[AM: #1]]}}}

\newcommand{\todo}[1]{\noindent{\textcolor{red}{{\textbf{[TODO]}: #1}}}}
\newcommand{\cut}[1]{}

\newcommand{\squishlist}{
	\begin{list}{$\bullet$}
		{
			\setlength{\itemsep}{0pt}
			\setlength{\parsep}{3pt}
			\setlength{\topsep}{3pt}
			\setlength{\partopsep}{0pt}
			\setlength{\leftmargin}{1.5em}
			\setlength{\labelwidth}{1em}
			\setlength{\labelsep}{0.5em} } }
	
	\newcommand{\squishend}{
\end{list}  }

\makeatletter
\def\@copyrightspace{\relax}
\def\@mkbibcitation{\relax}
\makeatother
\begin{document}


\title{\sysname: Efficient SQL Query Processing in Differentially Private Data Federations}




%
%
%
%

\numberofauthors{5} 

\author{
%
%
\alignauthor Johes Bater\\
\affaddr{Northwestern University}\\
%
\alignauthor Xi He\\
\affaddr{Duke University}\\
%
\alignauthor William Ehrich\\
\affaddr{Northwestern University}\\
%
\and
\alignauthor Ashwin Machanavajjhala\\
\affaddr{Duke University}\\
%
\alignauthor Jennie Rogers\\
\affaddr{Northwestern University}\\
}


\maketitle

\begin{abstract}
A private data federation is a set of autonomous databases that share a unified query interface offering \textit{in-situ} evaluation of SQL queries over the union of the sensitive data of its members. Owing to privacy concerns, these systems do not have a trusted data collector that can see all their data and their member databases cannot learn about individual records of other engines. Federations currently achieve this goal by evaluating queries obliviously using secure multiparty computation.  This hides the intermediate result cardinality of each query operator by exhaustively padding it. With cascades of such operators, this padding accumulates to a blow-up in the output size of each operator and a proportional loss in query performance. Hence, existing private data federations do not scale well to complex SQL queries over large datasets. 

We introduce \sysname, a private data federation that  offers data owners a differentially private view of the data held by others to improve their performance over oblivious query processing. \sysname uses computational differential privacy to minimize the padding of intermediate query results, achieving up to 35X performance improvement over oblivious query processing. When the query needs differentially private output, \sysname provides a trade-off between result accuracy and  query evaluation performance. 
\end{abstract}

\section{Introduction}\label{sec:intro}
The storage and analysis of data has seen dramatic growth in recent years. Organizations have never valued data more highly. Unfortunately, this newfound value has attracted unwanted attention. Security and privacy breaches litter the news headlines, and have engendered fear and a hesitance to share data, even among trusted collaborators. Without data sharing, information becomes siloed and enormous potential analytical value is lost.
 
\change{Recent work in databases and cryptography attempts to solve the data sharing problem by introducing the \textit{\federation}~\cite{Bater2017}. A \federation consists of a set of \cut{honest but curious} data owners who support a common relational database schema. Each party holds a horizontal partition (i.e., a subset of rows) of each table in the database. A \federation provides a SQL query interface to analysts (clients) over the union of the records held by the data owners. Query evaluation is performed securely from multiple data owners without revealing unauthorized information to any party involved in the query without the assistance of a trusted data curator.

A {\federation} must provably ensure the following guarantee: a data owner should not be able to reconstruct the database (or a part of it) held by other data owners based on the intermediate result sizes of the query evaluation. One way to achieve this without a trusted data curator is to use \mpc protocols, which provide a strong privacy guarantee that intermediate results leak \emph{no information} about their inputs. 

%

Current implementations of {\federation}s that use \mpc have a performance problem: state-of-the-art systems have slowdowns of \textit{4--6 orders of magnitude} over non-private data federations~\cite{Bater2017,Zheng2017}.
The source of this slowdown lies in the \mpc protocols used by {\federation}s that permit the secure evaluation of a SQL query without the need of trusted third party. A SQL query is executed as a directed acyclic graph of database operators. Each operator within the tree takes one or two inputs, applies a function, and outputs its result to the next operator. In a typical (distributed) database engine, the execution time needed to calculate each intermediate result and the size of that result leaks information about the underlying data to curious data owners participating in the \mpc. To address this leakage, {\federation}s insert dummy tuples to pad intermediate results to their maximum possible size thereby ensuring that execution time is independent of the input data. \cut{We discuss the details of our implementation of this in Section~\ref{sec:design}.} With secure evaluation, query execution effectively runs in the worst-case, drastically increasing its computation costs as intermediate result sizes grow. While performance is reasonable for simple queries and small workloads, performance is untenable for complex SQL queries with multiple joins.}

\change{Several approaches attempt to solve this performance problem, but they fail to provably bound the information leaked to a data owner.} One line of research uses Trusted Execution Environments (TEEs) that evaluate relational operators within an on-chip secure hardware enclave~\cite{Xu2018,Zheng2017}. TEEs efficiently protect query execution, but they require specialized hardware from chip manufacturers. Moreover, current TEE implementations from both Intel and AMD do not adequately obscure computation, allowing observers to obtain supposedly secure data through widely publicized hardware vulnerabilities~\cite{Chen2018,Ferraiuolo2017,Kocher2018,Lipp2016,VanBulck2018}. Another approach selectively applies homomorphic encryption to evaluate relational operators in query trees, while keeping the underlying tuples encrypted throughout the computation (e.g., CryptDB~\cite{Popa2011}). While this system improves performance, it leaks too much information about the data, such as statistics and memory access patterns, allowing a curious observer to deduce information about the true values of encrypted tuples~\cite{safeHarbor,Islam2012,naveed2015inference,taoInference18}.


\change{In this work, we bridge the gap between provable secure systems with untenable performance and practical systems with no provable guarantees on leakage using \emph{differential privacy} \cite{pdtextbook14}}. Differential privacy is a state-of-the-art technique to ensure privacy, and provides a provable guarantee that one can not reconstruct records in a database based on outputs of a differentially private algorithm. Differentially private algorithms, nevertheless, permit approximate aggregate statistics about the dataset to be revealed.

We present \sysname, a system that improves \federation performance by carefully relaxing the privacy guaranteed for data owners in terms of differential privacy. Instead of exhaustively padding intermediate results to their worst-case sizes, \sysname obliviously eliminates dummy tuples according to tunable privacy parameters, reducing each intermediate cardinality to a new, differentially private value. The differentially private intermediate result sizes are close to the true sizes, and thus, \sysname achieves practical query performance. 

\cut{Using differentially private intermediate sizes strongly, and non-linearly, affects the runtime of subsequent operators. Hence, we develop a novel I$/$O cost model to approximate the cost of the computation overhead and capture this size-dependent cost behavior. When carrying out the cost estimation and query optimization, our system automatically tunes the privacy parameters for each intermediate cardinality according to our cost model, maximizing performance and optimizing the differential privacy guarantees for each operator.

When the final query results themselves are differentially private, \sysname  adjusts the guarantees used during query execution to provide tighter accuracy bounds for the query output, trading off query evaluation performance for final result accuracy. In lieu of the all-or-nothing approach to privacy of current {\federation}s, we tune privacy parameters to control the amount of noise that we add and the trade-off we make between privacy, accuracy, and performance.}

\change{
To the best of our knowledge, \sysname is the first system for private data sharing that combines differential privacy with \mpc for query performance optimization. The main technical contributions in this work are:

\squishlist
    \item A query processing engine that offers controlled information leakage with differential privacy guarantees to speed up \federation query processing and provides tunable privacy parameters to control the trade-off between privacy, accuracy, and performance
    \item A computational differential privacy mechanism that securely executes relational operators, while minimizing intermediate result padding of operator outputs 
    \item A novel algorithm that optimally allocates, tracks, and applies differential privacy across query execution  
    \item A protocol-agnostic cost model that approximates the large, and non-linear, computation overhead of \mpc as it cascades up an operator tree  
\squishend

The rest of this paper is organized as follows. In Section~\ref{sec:pdfs} we define {\federation}s, outline our privacy goals and formally define secure computation and differential privacy. Section~\ref{sec:problem} describes the problem addressed by \sysname. Our end-to-end solution, \sysname, and its privacy guarantees, is described in Section~\ref{sec:design}. We show how to optimize the performance of \sysname in Section~\ref{sec:optimization}. Section~\ref{sec:implementation} describes how we implement specific \mpc protocols on top of {\sysname}'s protocol agnostic design. We experimentally evaluate our system implementation over real-world medical data in Section~\ref{sec:results}. We conclude with a survey of related work and future directions.
}

\section{Private Data Federation} \label{sec:pdfs}

\newcommand{\negligible}{\mathrm{negl}}
\newcommand{\view}{\mathrm{VIEW}}
\newcommand{\out}{\mathrm{OUTPUT}}
\newcommand{\dataowner}{\mathrm{DO}}
\newtheorem{assumption}{Assumption}

In this section, we formally define a \federation (PDF), describe privacy goals and assumptions, and define two security primitives -- secure computation and differential privacy.

\begin{figure} [t!]
	\centering
	\includegraphics[width=0.4\textwidth]{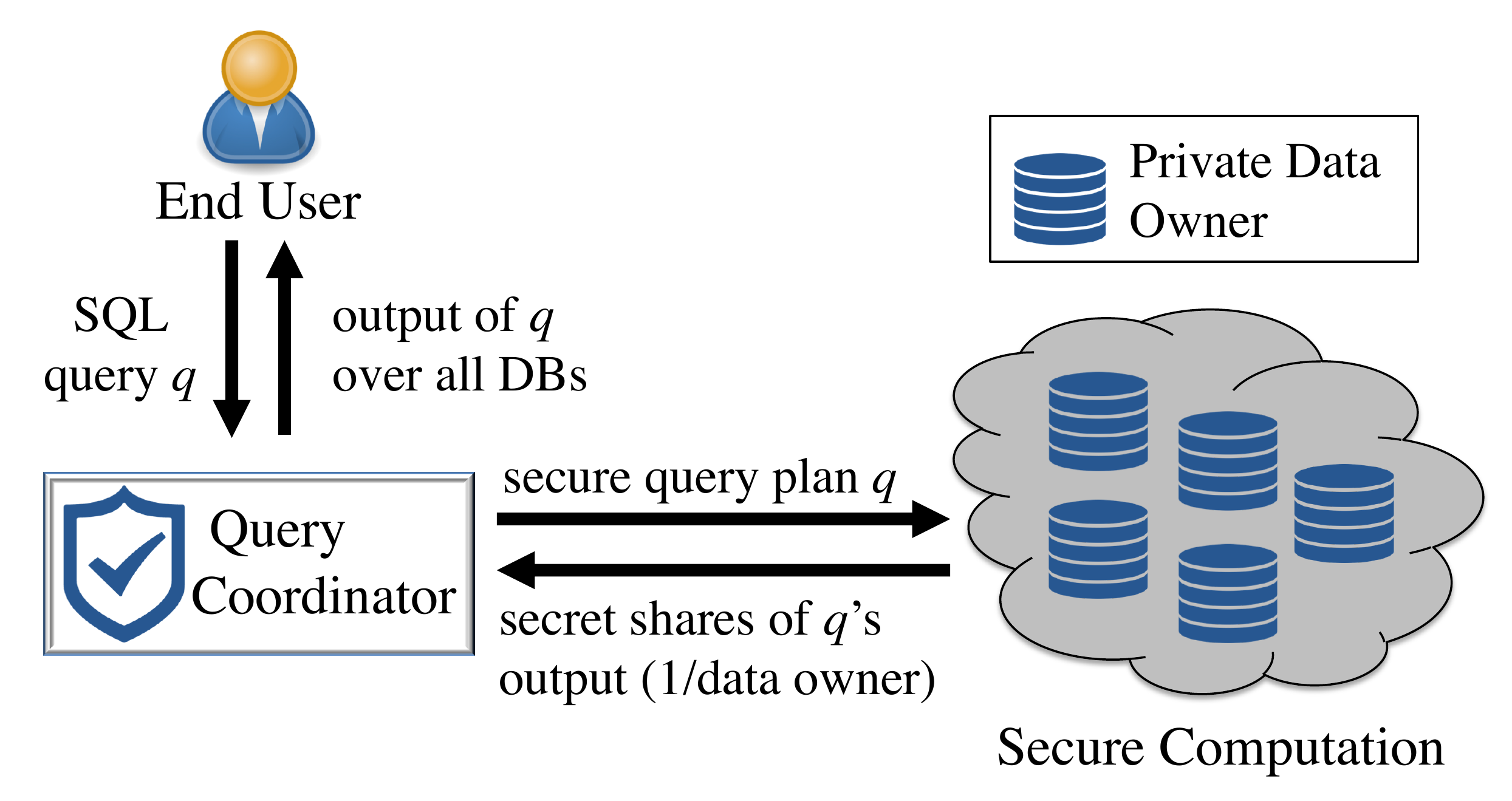}
	\caption{\change{Private data federation architecture}}
	\label{fig:pdf-architecture}
\end{figure}

	A \federation is a collection of autonomous databases that share a unified query interface for offering \textit{in-situ} evaluation of SQL queries over the union of the sensitive data of its members without revealing unauthorized information to any party involved in the query. A \federation has three types of parties: 1) two or more data owners $\dataowner_1,\ldots,\dataowner_m$ that hold private data $D_1,\ldots,D_m$ respectively, and where all $D_i$ share a public schema of $k$ relations ($R^1,\ldots,R^k$); 2) a \hb that plans and orchestrates SQL queries $q$ on the data of the data owners; and 3) a client that writes SQL statements to the \hb. The set of private data $\bar{D}=(D_1,\ldots,D_m)$ owned by the data owners is a horizontal partition of every table in the total data set $D$.
	
	As shown in Figure~\ref{fig:pdf-architecture}, the client first passes an SQL statement $q(\cdot)$ to the \hb, who compiles the query into an optimized secure query plan to be executed by each data owner. Once the data owners finish execution, they each pass their secret share of the output to the \hb, who assembles and returns the result to the  client.

\subsection{Privacy Goals and Assumptions}\label{sec:privacygoals}
Table~\ref{table:trustmodel} summarizes the privacy goals for all parties involved in this process and the necessary assumptions required for achieving these goals.

	\stitle{\bf Privacy Goals:} Data owners require a privacy guarantee that the other parties (data owners, \hb and the client) are not able to reconstruct the private data they hold based on intermediate results of the query execution. In particular, we require data owners to have a differentially private view over the inputs of other data owners (formally defined in Def.~\ref{def:dp}). We aim to support two policies for clients. Clients may be trusted (output policy 1) and allowed to see true answers to queries but must not learn any other information about the private inputs held by data owners. Clients may be untrusted (output policy 2), in which case they only are permitted a differentially private view over the inputs. The \hb is an extension of the client and has the same privacy requirements.

	\begin{table}[t]
		\centering
		\resizebox{\columnwidth}{!}{%
			\begin{tabular}{|c|c|c|}
				\hline
				& \textbf{\begin{tabular}[c]{@{}c@{}}Client Sees \\ True Answers\end{tabular}} & \textbf{\begin{tabular}[c]{@{}c@{}}Client Sees \\ Noisy Answers\end{tabular}}            \\ \hline
				\textbf{\begin{tabular}[c]{@{}c@{}}Privacy Goal \\ for DOs\end{tabular}}    & \multicolumn{2}{c|}{\begin{tabular}[c]{@{}c@{}} Differentially private (DP) view over \\ data held by other DOs\end{tabular}}                                                       \\
				\hline
				\textbf{\begin{tabular}[c]{@{}c@{}}Privacy Goal \\ for Client\end{tabular}} & \begin{tabular}[c]{@{}c@{}}Learns only \\ query answer\end{tabular}   & \begin{tabular}[c]{@{}c@{}}DP view over \\ data held by DOs\end{tabular} \\ \hline
				\textbf{\begin{tabular}[c]{@{}c@{}}Client \& DO\\ Assumption\end{tabular}}        & \multicolumn{2}{c|}{Semi-honest \& Computationally Bounded}                                                                   \\ \hline
				\textbf{\begin{tabular}[c]{@{}c@{}}DOs Colludes \\ with Client\end{tabular}}  & Not allowed                                                                  & Allowed             \\ \hline
				\textbf{\begin{tabular}[c]{@{}c@{}}Client Colludes \\ with some DOs\end{tabular}}  & Not allowed                                                                  & Allowed             \\ \hline
			\end{tabular}%
		}
		\caption{\change{Privacy Goals and Assumptions}} \label{table:trustmodel}
		
	\end{table}
	
	\stitle{Assumptions:} To achieve the privacy goals presented above, a \federation assumes that all parties are computationally bounded and work in the semi-honest setting. Hence, the \hb honestly follows the protocols and creates a secure plan to evaluate the query. All parties faithfully execute the cryptographic protocol created by the \hb, but may attempt to infer properties of private inputs held by other parties by observing the query instruction traces and data access patterns. The \hb is also assumed to be memory-less as it destroys its contents after sending query results back to the client. When the client/\hb sees the true answers (output policy 1), we assume that the client/\hb is trusted to not collude with data owners; otherwise, if the client/\hb shares the true answer with a data owner (or publishes the true answer), the data owner who colludes with the client/\hb can use his own private data and the true query answer to infer the input of other data owners. Similarly, we assume the data owners do not share their true input with the client/\hb; otherwise, the client/\hb can infer the other data owners' input and gain more information than the query answer. When the client/\hb is not trusted, we assume that the client may collude with (as many as all but one) data owners. Our guarantees hold even in that extreme case.

\subsection{Security Primitives}
\stitle{Secure Computation:}
To securely compute functions on distributed data owned by multiple parties, \mpc primitives are required. Let $f:\mathcal{D}^m\rightarrow \mathcal{R}^m$ be a functionality\footnote{Functionality $f$ can be both deterministic or randomized} over $\bar{D}=(D_1,\ldots,D_m)\in \mathcal{D}^m$, where $f_i(\bar{D})$ denotes the $i$-th element of $q(\bar{D})$. Let $\Pi$ be an $m$-party protocol for computing $f$ and $\view_i^{\Pi}(\bar{D})$ be the view of the $i$-th party during an execution of $\Pi$ on $\bar{D}$. For $I=\{i_1,\ldots,i_t\}\subseteq [m]$, we let
$D_I=(D_{i_1},\ldots,D_{i_t})$, $f_I(\bar{D})=(q_{i_1}(\bar{D}),\ldots,q_{i_t}(\bar{D}))$,  and $\view_I^{\Pi}(\bar{D})=(I, \view_{i_1}^{\Pi}(\bar{D}),\ldots,$\\ $\view_{i_1}^{\Pi}(\bar{D}))$.
We define \mpc with respect to semi-honest behavior~\cite{Goldreich:2004:FCV:975541} as follows.
\begin{definition}[\mpc]\label{def:mpc}
We say $m$-party protocol securely computes $f:\mathcal{D}^m\rightarrow \mathcal{R}^m$ if there exists a probabilistic polynomial-time algorithm denoted $S$ and $\negligible(\kappa)$  such that for any non-uniform probabilistic polynomial adversary $A$ and for every $I\subseteq [m]$, every polynomial $p(\cdot)$, every sufficiently large $\kappa \in \mathbb{N}$, every $\hat{D}\in \mathcal{D}^m$ of size at most $p(\kappa)$, it holds that
\begin{eqnarray}
|\Pr[A(S(I,D_I,f_I(\bar{D})),f(\bar{D}))=1]- & \\ \nonumber
\Pr[A(\view_I^{\Pi}(\bar{D}),\out^{\Pi}(\bar{D}))=1]|&\leq \negligible(\kappa),
\end{eqnarray}
where $\negligible(\kappa)$ refers to any function that is $o(\kappa^{-c})$, for all constants $c$ and $\out^{\Pi}(\bar{D})$ denotes the output sequence of all parities during the execution represented in $\view^{\Pi}_I(\bar{D})$.
\end{definition}
Various cryptographic protocols that enable \mpc have been studied~\cite{Beaver1990,chow2009two,Clifton2002,Goldreich:oram,Saia2014,Liu2015,Wang2017,Yao1982} and several are applicable to relational operators~\cite{Bater2017,bogdanov2014privacy,conclave}.

\change{
\stitle{Differential Privacy:}
Differential privacy~\cite{pdtextbook14} is an appealing choice to bound the information leakage on the individual records in the input while allowing multiple releases of statistics about the data. This privacy notion is utilized by numerous organizations, including the US Census Bureau~\cite{onthemap08}, Google~\cite{erlingsson2014rappor}, and Uber~\cite{practicaldp18}. Formally, differential privacy is defined as follows.
}

\begin{definition}[$(\epsilon,\delta)$-Differential Privacy (DP)]\label{def:dp}
A randomized mechanism $f:\mathcal{D}\rightarrow \mathcal{R}$ satisfies $(\epsilon,\delta)$-differential privacy if for any pair of neighboring databases $D,D'\in \mathcal{D}$ such that $D$ and $D'$ differ by adding or removing a row and any set $O\subseteq \mathcal{R}$,
\begin{equation}
\Pr[f(D)\in O] \leq e^{\epsilon} \Pr[f(D')\in O] +\delta.
\end{equation}
\end{definition}

\change{
The differential privacy guarantee degrades gracefully when invoked multiple times. In the simplest form, the overall privacy loss of multiple differentially private mechanisms can be bounded with sequential composition~\cite{Dwork2006a}.
\begin{theorem}[Sequential Composition]\label{theorem:seqcomp}
	If $M_1$ and $M_2$ are $(\epsilon_1,\delta_1)$-DP and $(\epsilon_2,\delta_2)$-DP algorithms that use independent randomness, then releasing the outputs $M_1(D)$ and $M_2(D)$ on  database $D$ satisfies $(\epsilon_1+\epsilon_2,\delta_1+\delta_2)$-DP.
\end{theorem}
There exist advanced composition theorems that give tighter bounds on privacy loss under certain conditions \cite{pdtextbook14}, but we use sequential composition as defined above as it is general and easy to compute. 

\eat{
 \am{Maybe keep this for now.}
Many techniques that satisfy differential privacy use the following notion of query sensitivity:

\begin{definition}[Query Sensitivity]\label{def:sens}
The sensitivity of a function $f:\mathcal{D}\rightarrow \mathbb{R}^*$ is the maximum difference in the functions output for any pairs of neighboring databases $D$ and $D'$,
\begin{displaymath}
\Delta f = \max_{|D-D'|\cup |D'-D| = 1} \|f(D)-f(D')\|_1
\end{displaymath}
\end{definition}
A classic mechanism that satisfies $(\epsilon,0)$-differential privacy is the Laplace mechanism~\cite{pdtextbook14} defined as follows:
\begin{definition}[Laplace Mechanism]
Given a function function $f:\mathcal{D}\rightarrow \mathbb{R}^d$, the Laplace mechanism adds to true answer $f(D)$ a vector of independent noise variables $\eta\in \mathbb{R}^d$ drawn from the Laplace distribution $Lap(0,\Delta f/\epsilon)^d$, i.e. $\Pr[\eta_i=x]\propto e^{-x\cdot\epsilon/\Delta q}$ for $i=1,\ldots,d$.
\end{definition}
}
}

\section{Problem Statement} \label{sec:problem}
Our goal is to build a system that permits efficient query answering on {\federation}s while achieving provable guarantees on the information disclosed to the clients as well as the data owners. As discussed in Section~\ref{sec:intro}, prior work can be divided into two classes: (i) systems that use \mpc to answer queries, which give formal privacy guarantees but suffer \textit{4--6 orders of magnitude} slowdowns compared to non-private data federations ~\cite{Bater2017,Zheng2017}, and (ii) systems that use secure hardware \cite{Xu2018,Zheng2017} or property preserving encryption \cite{Popa2011}, which are practical in terms of performance, but have no formal guarantee on the privacy afforded to clients and data owners.


While \mpc based solutions are attractive, and there is much recent work in improving their efficiency~\cite{SP:BHKR13,USS:HEKM11,C:IKNP03}, \mpc solutions fundamentally limit performance due to several factors: 1) a worst-case running time is necessary to avoid leaking information from early termination; 2) computation must be \textit{oblivious}, i.e., the program counters and the memory traces must be independent of the input; 3) cryptographic keys must be sent and computed on by data owners, the cost of which scales with the complexity of the program.

In particular, applying \mpc protocols to execute SQL queries results in extremely slow performance. A SQL query is a directed acyclic graph of database operators. Each operator within the tree takes an input, applies a function, and outputs its result to the next operator. The execution time needed to calculate each intermediate result and the array size required to hold that result leaks information about the underlying data. To address this leakage, {\federation}s must insert dummy tuples to pad intermediate results to their maximum possible size and ensure that execution time is independent of the input data. We demonstrate this property in the following example.

\begin{example}\label{eg:exhaustive-padding}
Figure~\ref{fig:exhaustive-padding} shows a SQL query and its corresponding operator tree. Tuples flow from the source relations through a filter operator, before entering two successive join operators and ending at a distinct operator. In order to hide the selectivity of each operator, the {\federation} pads each intermediate result to its maximum value. At the filter operators, the results are padded as if no tuples were filtered out, remaining at size $n$. This is a significant source of additional computation. If the selectivity of the filter is $10^{-3}$, our padding gives a $1000\times$ performance overhead.  As the tuples flow into the first join, which now receives two inputs of size $n$, exhaustive padding produces an output of $n^2$ tuples. When this result passes through the next join, the result size that the distinct operator must process jumps to $n^3$. If we once again assume a $10^{-3}$ selectivity at the joins, we now see a $10^9\times$ computation overhead. Such a significant slowdown arises from the cumulative effect of protecting cascading operators. By exhaustively padding operator outputs, we see ever-increasing intermediate result cardinalities and, as a function of that growth, ever-decreasing performance.
\end{example}

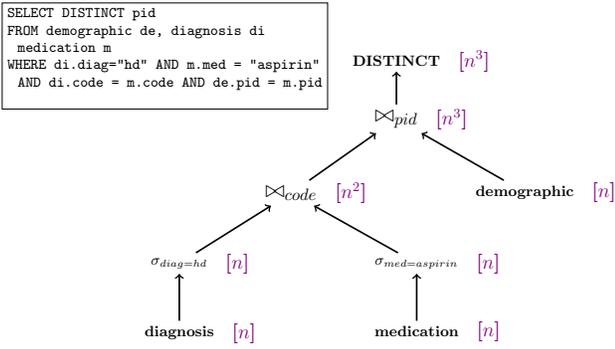
\begin{figure} [!t]
\centering
\resizebox {\columnwidth} {!} {
\begin{tikzpicture}
\node[text centered, font=\bfseries] (distinct) {DISTINCT};
\node[below = 0.7 of distinct, text centered, font=\bfseries] (djoin) {\begin{Large}$\bowtie_{pid}$\end{Large}};

\node[below right = of djoin, text centered, font=\bfseries] (demo) {demographic};

\node[below left = of djoin, text centered,font=\bfseries] (join) {\begin{Large}$\bowtie_{code}$\end{Large}};
\node[below left =  of join, text centered, font=\bfseries] (dfilter) {$\sigma_{diag=hd}$  };
\node[below right =   of join, text centered,  font=\bfseries] (mfilter) {$\sigma_{med=aspirin}$ };

\node[below =  of dfilter, text centered, font=\bfseries] (dscan) {diagnosis };
\node[below= of mfilter, text centered, font=\bfseries] (mscan) {medication};

\node[right= 0.15 of distinct, font=\bfseries] (distinct_card) {\color{violet} \large [$n^3$]};
\node[right= 0.15 of djoin, font=\bfseries] (djoin_card) {\color{violet} \large [$n^3$]};
\node[right= 0.15 of join, font=\bfseries] (join_card) {\color{violet} \large [$n^2$]};

\node[right= 0.15 of demo, font=\bfseries] (demo_card) {\color{violet} \large [$n$]};
\node[right= 0.15 of dscan, font=\bfseries] (dscan_card) {\color{violet} \large [$n$]};
\node[right= 0.15 of mscan, font=\bfseries] (mscan_card) {\color{violet} \large [$n$]};

\node[right= 0.15 of mfilter, font=\bfseries] (mfilter_card) {\color{violet} \large [$n$]};

\node[right= 0.15 of dfilter, font=\bfseries] (dfilter_card) {\color{violet} \large [$n$]};


\draw[->, line width= 1] (demo)  --  (djoin);
\draw[->, line width= 1] (djoin)  --  (distinct);
\draw[->, line width= 1] (join)  --  (djoin);

\draw[->, line width= 1] (dfilter)  --  (join);

\draw[->, line width= 1] (mfilter)  --  (join);
\draw[->, line width= 1] (dscan)  --  (dfilter);

\draw[->, line width= 1] (mscan)  --  (mfilter);

\node[rectangle, draw, align=left, above left = -1.25 and 2.5 of distinct.north east] (query) {
{\normalsize\tt SELECT DISTINCT pid}\\
{\normalsize\tt FROM demographic de, diagnosis di}\\
{\normalsize\tt \,\,\,\,medication m}\\
{\normalsize\tt WHERE di.diag="hd" AND m.med = "aspirin"}\\
{\normalsize\tt \,\,\,\,AND di.code = m.code AND de.pid = m.pid}\\
};
\end{tikzpicture}
}
\caption{Exhaustive padding of intermediate results in an oblivious query evaluation}
\label{fig:exhaustive-padding}
\vspace{-6mm}
\end{figure}


To tackle this fundamental performance challenge, we formalize our research problems as follows.
\squishlist
\item Build a system that uses differentially private leakage of intermediate results to speed up SQL query evaluation while achieving all the privacy goals stated in Section~\ref{sec:pdfs}.
\item Given an SQL operator $o$ and a privacy budget $(\epsilon, \delta)$, design an $(\epsilon, \delta)$-differentially private mechanism that executes the operator with a smaller performance overhead compared to fully oblivious computation.
\item Given a sequence of operators $Q=\{o_1,\ldots,o_l\}$, design an algorithm to optimally split the privacy budget $(\epsilon,\delta)$ among these operators during query execution.
\squishend

In the following sections, we first present our system for more efficient SQL query processing and its building blocks including the differentially private mechanisms for each SQL operator and the overall privacy analysis (Section~\ref{sec:design}), and then present a budget splitting algorithm that optimizes the performance (Section~\ref{sec:optimization}).

\section{SHRINKWRAP}
\label{sec:design}
\change{
In this section, we present the end-to-end algorithm of \sysname for answering a SQL query in a \federation, the key differentially private mechanism for safely revealing the intermediate cardinality of each operator, and the overall privacy analysis of \sysname.
}

\subsection{End-to-end Algorithm}
\change{
At its core, \sysname is a system that applies differential privacy throughout SQL query execution to reduce intermediate result sizes and improve performance.  The {\federation} ingests a query $q$ as a directed acyclic graph (DAG) of database operators $Q=\{o_1,\ldots,o_l\}$, following the order of a bottom up tree traversal. The \hb decides the output policy based on the client type and assigns the privacy parameters $(\epsilon,\delta)$ for the corresponding privacy goals. To improve performance, a privacy budget of $(\epsilonper,\deltaper)$ will be spent on protecting intermediate cardinalities of the $l$ operators. When client is allowed to learn the true query output (output policy 1), then $\epsilonper=\epsilon$ and $\deltaper=\delta$; when the client is allowed to learn noisy query answers from a differentially private mechanism, then $\epsilonper <\epsilon$ and $\deltaper <\delta$, where the remaining privacy budget $(\epsilon-\epsilonper,\delta-\deltaper)$ will be spent on the query output.

Algorithm~\ref{algo:sw} outlines the end-to-end execution of a SQL query in \sysname. Inputs to the algorithm are the query $Q$, the overall privacy parameters $(\epsilon,\delta)$, the privacy budget for improving the performance $(\epsilonper,\deltaper)$, and the public information $K$, including the relational database schema and the maximum possible input data size at each party.

First, the \hb allocates the privacy budget for improving performance $(\epsilonper,\deltaper)$ among the operators in $Q$ such that 
\vspace{-1em}
\begin{equation}\label{eq:budgetsplit}
\sum_{i=1}^l \epsilon_i = \epsilonper \mbox{ and } \sum_{i=1}^l \delta_i=\deltaper
\end{equation} 
where $(\epsilon_i, \delta_i)$ is the privacy budget allocated to operator $o_i$. If $\epsilon_i = \delta_i = 0$, then the operator is evaluated obliviously (without revealing the intermediate result size). When $\epsilon_i, \delta_i > 0$, \sysname is allowed to reveal an overestimate of the intermediate result size, with larger $(\epsilon_i, \delta_i)$ values giving tighter intermediate result estimates.  The performance improvement of \sysname highly depends on how the budget is split among the operators, but any budget split that satisfies Equation~\ref{eq:budgetsplit} satisfies privacy. We show several allocation strategies in Section~\ref{sec:optimization} and analyze them in the evaluation.


Next, the \hb compiles a secure query plan to be executed by the data owners. The secure query plan traverses the operators in $Q$. For each operator $o_i$, data owners compute secret shares of the true output using \mpc over the federated database and the output of other operators computed from $D$. We denote this secure computation for $o_i$ by $SMC(o_i,D)$. The true outputs (of size $c_i)$ are placed into a secure array $O_i$. The secure array is padded (with encrypted dummy values) so that its size equals the maximum possible output size for operator $o_i$.

 
Then, \sysname calls a $(\epsilon_i,\delta_i)$-differentially private function $Resize()$ to resize this secure array $O_i$ to a smaller array $S_i$ such that a subset of the (encrypted) dummy records are removed. Resizing is described in Section~\ref{sec:resize}. Once all the operators are securely evaluated, if the output policy allows release of true answers to the client (policy 1), the last secure array $S_l$ will be directly returned to the \hb; otherwise,  the data owners jointly and securely compute a differentially private mechanism on $S_l$ with the remaining privacy budget $(\epsilon_0=\epsilon-\epsilonper, \delta_0=\delta-\deltaper)$, denoted by $M^{(\epsilon_0,\delta_0)}(S_l)$, and return the output of $M^{(\epsilon_0,\delta_0)}(S_l)$ to the \hb. For instance, one could use Laplace mechanism to perturb count queries using a multiparty protocol for generating noise as in~\cite{Narayan2012}. Finally, the \hb assembles the final secure array received from the data owners $S_l$, decrypting the final result $R$ and returning it to the client. 
}

\setlength{\textfloatsep}{0.1cm}
\RestyleAlgo{boxruled}
\begin{algorithm}[t]
{\small
\DontPrintSemicolon
\KwIn{PDF query DAG with operators $Q=\{o_1,\ldots,o_l\}$,
the federated database $D$,
public information $K$,
privacy parameters $(\epsilon,\delta)$,
privacy budget for performance $(\epsilonper\leq \epsilon,\deltaper\leq \delta)$
}
\KwOut{Query output $R$}
\BlankLine
$//$\emph{Computed on \hb}\;
$P=\{(\epsilon_1,\delta_1),\ldots, (\epsilon_l,\delta_l) | \sum_{i=1}^l \epsilon_i = \epsilonper,\sum_{i=1}^l \delta_i = \deltaper\}$
$\leftarrow$ AssignBudget$(\epsilonper,\deltaper,Q,K)$ $//$ \emph{(Sec.~\ref{sec:optimization})} \;
\BlankLine
$//$\emph{Computed on private data owners}\;
\For{$i\leftarrow 1$ \KwTo $l$}{
    $(O_i,c_i) \leftarrow$  SMC($o_i,D$) $//$\emph{with exhaustive padding}\;
    $S_i \leftarrow $  Resize($O_i$,$c_i$,$\epsilon_i$,$\delta_i$,$\Delta \carquery_i$)  $//$ \emph{DP resizing (Sec.~\ref{sec:resize})}\;
}
\change{
\uIf{$\epsilon=\epsilonper$ and $\delta=\deltaper$ }{
    Send \hb $S_l$ $//$\emph{output policy 1}\;
}
\Else{
    Query output budget $\epsilon_0 \leftarrow (\epsilon-\epsilonper)$, $\delta_0 \leftarrow (\delta-\deltaper)$  \;
    Send \hb $M^{(\epsilon_0,\delta_0)}(S_l)$ $//$\emph{output policy 2}\;
}}
\BlankLine
$//$\emph{Computed on \hb}\;
$R \leftarrow$ Assemble($S_l$) \;
return $R$ to client\;
\BlankLine
\BlankLine
\emph{function} Resize($O$, $c$, $\epsilon$, $\delta$, $\Delta \carquery$)\;
\Indp
$\tilde{c} \leftarrow c + $TLap($\epsilon$, $\delta$, $\Delta \carquery$) $//$ \emph{perturb cardinality} \;
$O \leftarrow$ Sort($O$) $//$\emph{obliviously sort dummy tuples to the end}\;
$S \leftarrow$ new SecureArray($O[1,\ldots,\tilde{c}]$) \;
return $S$\;
\BlankLine
\caption{End-to-end \sysname algorithm}
\label{algo:sw}
}
\end{algorithm}
\setlength{\floatsep}{0.1cm}


\change{
\subsection{DP Resizing Mechanism}\label{sec:resize}
For each operator $o_i$, \sysname first computes the true result using \mpc and places it into a exhaustively padded secure array $O_i$. For example, as shown in Figure~\ref{fig:padding}, a join operator with two inputs, each of size $n$, will have a secure array $O_i$ of size $n^2$. This is the exhaustive padding case shown on the left side of the figure, with the {\federation} inserting the materialized intermediate results into the secure array and padding those results with dummy tuples.  The Resize() function takes in the exhaustively padded secure array $O_i$, the true cardinality of the output $c_i$, the privacy parameters $(\epsilon_i,\delta_i)$, and the sensitivity of the cardinality query at operator $o_i$, denoted by $\Delta \carquery_i$. The cardinality query $\carquery_i(D)$ returns the output of the operator $o_i$ in the PDF query $q$. The sensitivity of a query is defined as follows.
\begin{definition}[Query Sensitivity]\label{def:sens}
The sensitivity of a function $f:\mathcal{D}\rightarrow \mathbb{R}^*$ is the maximum difference in its output for any pairs of neighboring databases $D$ and $D'$,
\begin{displaymath}
\Delta f = \max_{\{D, D' \mbox{ s.t. } |D-D'|\cup |D'-D| = 1\}} \|f(D)-f(D')\|_1
\end{displaymath}
\end{definition}

Based on these inputs, the Resize() function runs a $(\epsilon_i,\delta_i)$-differentially private truncated Laplace mechanism (Def.~\ref{def:truncatedlap}) to obtain a new, differentially private cardinality $\tilde{\carquery}_i > \carquery_i$ for the secure array. Next, this function sorts the exhaustively padded secure array $O_i$ obliviously such that all dummy tuples are at the end of the array and then copies the  first $\tilde{\carquery}_i$ tuples in the sorted array into a new secure array $S_i$ of size $\tilde{\carquery}_i$. This new secure array $S_i$ of a smaller size than $O_i$ is used for the following computations.


\begin{figure} [b]
\centering
\includegraphics[width=0.4\textwidth]{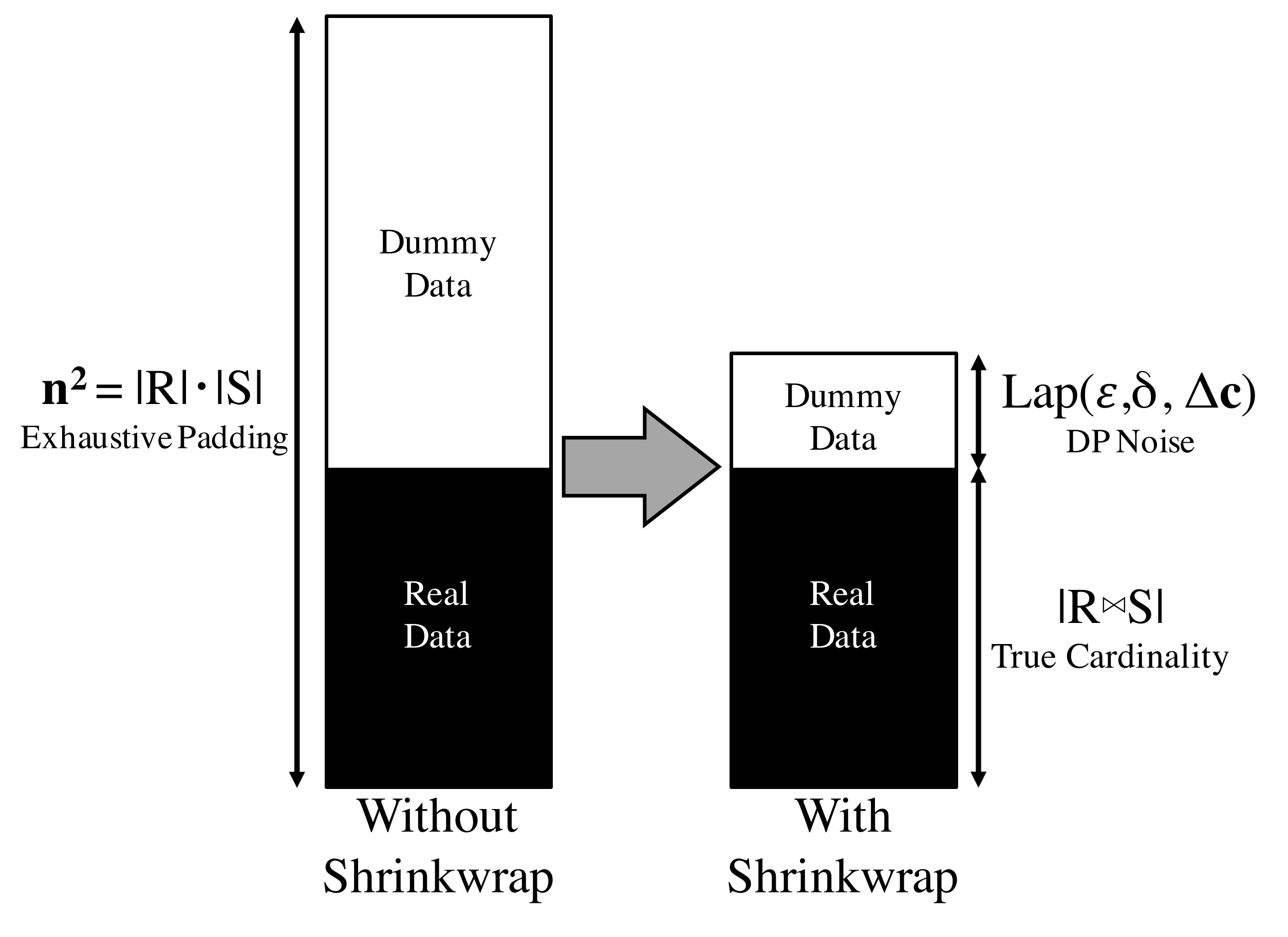} 
\caption{Effect of Shrinkwrap on Intermediate Result Sizes When Joining Tables R and S} \label{fig:padding}
\end{figure}

\stitle{Noise Generation:}
\sysname generates differentially private cardinalities by using a truncated Laplace mechanism. Given an operator $o$, with a privacy budget $(\epsilon, \delta)$ and the sensitivity parameter for the cardinality query $\Delta \carquery$, this mechanism computes a noisy cardinality for use in Algorithm~\ref{algo:sw}.

\begin{definition}[Truncated Laplace Mechanism]\label{def:truncatedlap}
\ \\
Given a query $\carquery:\mathcal{D} \rightarrow \mathbb{N}$, the truncated Laplace mechanism $\tlap(\epsilon,\delta,\Delta \carquery)$ adds a non-negative integer noise $\max(\eta,0)$ to the true query answer $\carquery(D)$, where $\eta \in \mathbb{Z}$ follows a distribution, denoted by $L(\epsilon,\delta,\Delta \carquery)$
that has a probability density function $\Pr[\eta=x] = p \cdot e^{- (\epsilon/\Delta \carquery)|x-\eta^0|}$,
where $p = \frac{e^{\epsilon/\Delta \carquery} - 1}{e^{\epsilon/\Delta \carquery}+ 1}$,
$\eta_0 = -\frac{\Delta \carquery \cdot \ln((e^{\epsilon/\Delta \carquery} + 1)\delta)}{\epsilon} + \Delta \carquery$.
\end{definition}
This mechanism allows us to add non-negative integer noise to the intermediate cardinality query which corresponding to the number of the dummy records. The noise $\eta$ drawn from $L(\epsilon,\delta,\Delta \carquery)$ satisfies that  $\Pr[\eta<\Delta \carquery] \leq \delta$, which enables us to show the privacy guarantee of this mechanism.

\begin{theorem} The truncated Laplace mechanism satisfies $(\epsilon,\delta)$-differential privacy. \end{theorem}
\begin{proof} 
For any neighboring databases $D_1,D_2$, let $c_1=\carquery(D_1)\geq 0$ and $c_2=\carquery(D_2)\geq 0$. W.L.O.G. we consider $c_2\geq c_1$. It is easy to see that $\Pr[\tlap(D_1) \in (-\infty,c_1)]=0$ and $\Pr[\tlap(D_2) \in (-\infty,c_2)]=0$.
By the noise property, for any $o\geq c_2\geq c_1$, it is true that $|\ln\frac{\Pr[\tlap(D_2)=o]}{\Pr[\tlap(D_1)=o]}|\leq\epsilon$. However, when the output $o \in [c_1, c_2)$, $\Pr[\tlap(D_2) = o] = 0$, but the $\Pr[\tlap(D_2) = o] > 0$, making the ratio of probabilities unbounded. Nevertheless, we can bound $\Pr[\tlap(D_2) \in [c_1, c_2)]$ by $\delta$ as shown below. 

Let $O^* = (-\infty, c_2)$. We can show that for any output set $O$ of query $\carquery()$, we have
\begin{eqnarray}
&& \Pr[\tlap(D_1)\in O] \nonumber\\
&=& \Pr[\tlap(D_1)\in (O\cap O^*)] + \Pr[\tlap(D_1)\in (O-O^*)]  \nonumber \\
&\leq& \Pr[\tlap(D_1) \in [c_1,c_2)] + e^{\epsilon}\Pr[\tlap(D_2)\in (O-O^*)] \nonumber \\
&=& \Pr[\eta_1 < \Delta \carquery] + e^{\epsilon}\Pr[\tlap(D_2)\in O] \nonumber\\
&=& \delta + e^{\epsilon}\Pr[\tlap(D_2)\in O]
\end{eqnarray}
Hence, this mechanism satisfies $(\epsilon,\delta)$-DP.
\end{proof}
}

\stitle{Sensitivity Computation:}
In order to create our noisy cardinalities using the truncated Laplace mechanism in Definition~\ref{def:truncatedlap}, we must determine the cardinality query sensitivity at the output of each operator in the query tree. An intermediate result cardinality, $\carquery_i$, is the output size of the sub-query tree rooted at operator $o_i$. The sensitivity of $\carquery_i$ depends on the operators in its sub-query tree, and we compute this as a function of the stability of operators in the sub-query tree. The stability of an operator bounds the change in the output with respect to the change in the input of the operator.
\begin{definition}[Stability~\cite{mcsherry2009privacy}]
A transformation operator is $s$-stable if for any two input data sets A and B,
\begin{displaymath}
|T(A) \ominus T(B)| \leq s \times |A \ominus B|
\end{displaymath}
\label{def:stability}
\vspace{-6mm}
\end{definition}
For unary operators, such as SELECT, PROJECT, the stability usually equals one. For JOIN operators, a binary operators, the stability is equal to the maximum multiplicity of the input data. The sensitivity of the intermediate cardinality query $\carquery_i$ can be recursively computed bottom up, starting from a single change in the input of the leaf nodes.
\begin{example} Given a query tree shown in Figure~\ref{fig:sensitivity-tree}, the bottom two filter operators each have a stability of one, giving a sensitivity ($\Delta \carquery$) of 1 at that point. Next, the first join has a stability of $m$ for the maximum input multiplicity, which increases the overall sensitivity to $m$. The following join operator also has a stability of $m$, which changes the sensitivity to $m^2$. Finally, the DISTINCT operator has a stability of 1, so the overall sensitivity remains $m^2$.
\begin{figure} [tph]
\centering
\resizebox {0.9\columnwidth} {!} {
\begin{tikzpicture}
\node[text centered, text width=3cm, font=\bfseries] (distinct) {DISTINCT \color{violet} \\ \large [$\Delta \carquery_5 = m^2$]};
\node[below = 0.7 of distinct, text centered, text width=3cm, font=\bfseries] (djoin) {\begin{Large}$\bowtie_{pid}$\end{Large} \\ \color{violet} \large [$\Delta \carquery_4 = m^2$]};
\node[below right = 0.5 of djoin, text centered, font=\bfseries] (demo) {demographics};
\node[below left =  0.5 of djoin, text centered, text width=3cm, font=\bfseries] (join) {\begin{Large}$\bowtie_{pid}$\end{Large} \\ \color{violet} \large [$\Delta \carquery_3 = m$]};
\node[below left =  0.5 of join, text centered, text width=3cm, font=\bfseries] (dfilter) {$\sigma_{diag=hd}$ \\ \color{violet} \large [$\Delta \carquery_1 = 1$]};
\node[below right =  0.5 of join, text centered, text width=3cm, font=\bfseries] (mfilter) {$\sigma_{med=aspirin}$ \\ \color{violet} \large [$\Delta \carquery_2 = 1$]};

\node[below = 0.9 of dfilter, text centered, font=\bfseries] (dscan) {diagnosis};
\node[below=0.7 of mfilter, text centered, font=\bfseries] (mscan) {medication};

\draw[->, line width= 1] (demo)  --  (djoin);
\draw[->, line width= 1] (djoin)  --  (distinct);
\draw[->, line width= 1] (join)  --  (djoin);

\draw[->, line width= 1] (dfilter)  --  (join);

\draw[->, line width= 1] (mfilter)  --  (join);
\draw[->, line width= 1] (dscan)  --  (dfilter);

\draw[->, line width= 1] (mscan)  --  (mfilter);

\end{tikzpicture}
}
\caption{Sensitivity analysis for running example} 
\label{fig:sensitivity-tree}
\vspace{-6mm}
\end{figure}
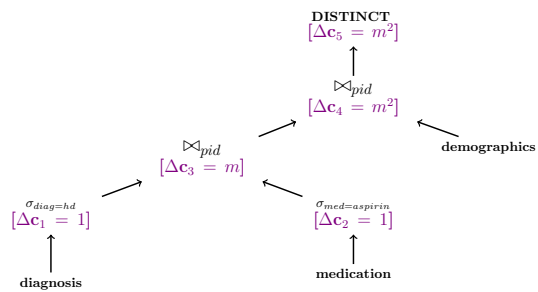
\end{example}
Using these principles, \sysname calculates the sensitivity for each operator in a query tree for use in Laplace noise generation during query execution. More details can be seen in PINQ \cite{mcsherry2009privacy} for computing the sensitivity of SQL queries.

\stitle{Secure Array Operations:}
Note that each time \sysname adds or removes tuples, it normally needs to pay an I$/$O cost for accessing the secure array. We avoid paying the full cost by carrying out a bulk unload and load. To start the $Resize()$ function, \sysname takes an intermediate result secure array and sorts it so that all dummy tuples are at the end of the secure array. Next, \sysname cuts off the end of the original secure array and copies the rest into the new, differentially-private sized secure array. Since the new secure array has a size guaranteed to be larger than the true cardinality, we know that cutting off the end will only remove dummy tuples. Thus, we avoid paying the full I$/$O cost for each tuple. However, we do pay an $O(n\log(n))$ cost for the initial sorting, as well as an $O(n)$ cost for bulk copying the tuples.  This algorithm is an extension of the one in~\cite{goodrich2011privacy} and variants of it are in~\cite{Nayak, nikolaenko2013privacy, Zheng2017}.

\change{
\subsection{Overall Privacy Analysis}
Given a PDF query DAG $q$ with operators $\{o_1,\ldots, o_l\}$ and a privacy parameter $(\epsilon,\delta)$, Algorithm~\ref{algo:sw} achieves the privacy goals stated in Section~\ref{sec:pdfs}.
\begin{theorem}\label{theorem:sw-algo}
Under the assumptions described in Section~\ref{sec:pdfs}, Algorithm~\ref{algo:sw} achieves the privacy goals: data owners have a $(\epsilon,\delta)$-computational differentially private view over the input data of other data owners; when true answers are returned to the client, the client learn nothing else; when noisy answers are returned the client, the client has a $(\epsilon,\delta)$-computational differentially private view over the input data of all the data owners.
\end{theorem}
\begin{proof}(sketch)
We prove that the view of data owners satisfies computational differential privacy \cite{Mironov:2009:CDP:1615970.1615981,Vadhan2017} by showing that the view of any data owner is computationally indistinguishable from its view in an ideal simulation that outputs \emph{only} the differentially private cardinalities of each operator output.

Let $Sim$ be a simulator that takes as input the horizontal partitions held by the data owners $(D_1, ..., D_m)$, and outputs a vector $S^\bot_i$ for operator $o_i$ such that: (a) $|S^\bot_i|$ is drawn from $\carquery_i + \tlap()$, and (b) $S^\bot_i$ contains encryptions of 0 (of the appropriate arity). It is easy to see that function computed by $Sim$ satisfies $(\epsilonper, \deltaper)$-DP: (a) at each operator, the release of $|S^\bot_i|$ satisfies $(\epsilon_i, \delta_i)$-DP, and (b) By sequential composition (Theorem~\ref{theorem:seqcomp}) and Equation~\ref{eq:budgetsplit}, releasing all the cardinalities satisfies $(\epsilonper, \deltaper)$-DP. When the output itself is differentially private, there is an additional privacy loss to data owners (since the client could release the output to the data owners), but the overall algorithm still satisfied $(\epsilon, \delta)$-DP by sequential composition.

The real execution outputs secure arrays $S_i$ at each operator $o_i$ using \mpc. This ensures that the view of a data owner in the real execution is computationally indistinguishable from the execution of $Sim$, which satisfies $(\epsilon,\delta)$-DP. Thus, we can show that the view of any data owner satisfies $(\epsilon, \delta)$-computational DP. The privacy guaranteed to the client and \hb can be analogously argued.
\end{proof}

\subsection{Multiple SQL Queries}\label{sec:multiplequeries}
The privacy loss of \sysname due to answering multiple queries is analyzed using sequential composition (Theorem~\ref{theorem:seqcomp}). There are two approaches to minimize privacy loss across multiple queries: (i) applying advanced composition theorems that give tighter bounds on privacy loss under certain conditions~\cite{pdtextbook14} or (ii) optimizing privacy budget allocation across the operators of a workload of predefined SQL queries (e.g. using \cite{McKenna:2018:OEH:3231751.3242939}). These two approaches are orthogonal to our work and adapting these ideas to the context of secure computation is an interesting avenue for future work.

}


\section{Performance Optimization}
\label{sec:optimization}
A key technical challenge is to divide the privacy budget for performance $(\epsilonper, \deltaper)$  across the different operators so the overall runtime of the query execution is minimized. 
\change{\begin{problem}\label{problem:research-problem}
  Given a privacy budget $(\epsilonper,\deltaper)$, a PDF query DAG $q$ with operators $Q=\{o_1,\ldots,o_l\}$, and public information $K$ which includes the sensitivities of the intermediate cardinality queries $\{\Delta \carquery_1,\ldots, \Delta \carquery_l\}$, the schema information, compute the share $(\epsilon_i, \delta_i)$ assigned to $o_i$ such that: (a) the assignment ensures privacy, i.e., $\sum_{i=1}^l \epsilon_i = \epsilonper$,  and (b) the overall run time of the query is minimized.
\end{problem}}


\subsection{Baseline Privacy Budget Allocation}
We first consider two baseline strategies for allocating the privacy budget to individual operators.

\stitle{Eager:} This approach allocates the entire budget $(\epsilonper,\deltaper)$ to the first operator in the query tree. In Figure~\ref{fig:exhaustive-padding}, this equates to minimizing the output cardinality of the filter on diagnoses and executing the rest of the tree obliviously. A benefit of this approach is that since the intermediate results of the first operator flow through the rest of the tree, every operator will see the benefit of minimizing the output cardinality. On the other hand, in some trees, certain operators, like joins, have an output-sized effect on the total query performance. Not allocating any privacy budget to those operators results in a missed optimization opportunity.

\stitle{Uniform:} The second approach splits the privacy budget $(\epsilonper,\deltaper)$ uniformly across the query tree, resulting in equal privacy parameters for each operator. The benefit here is that every single operator receives part of the budget, so we do not lose out on performance gains from any operator. The drawback of the uniform approach, like the one-time approach, is that some operators produce larger effects, so by not allocating more of the budget to them, we lose out on potential performance gains.

\subsection{Optimal Privacy Budget Allocation}
We design a third approach which uses an execution cost model as an objective function and applies convex optimization to determine the optimal budget splitting strategy.

\stitle{Cost Model:} The execution cost of a query in Algorithm~\ref{algo:sw} mainly consists of (a) the cost for securely computing each operator and (b) the I$/$O cost of handling each intermediate result by \sysname. We model these cost as functions of the corresponding input data sizes and output data size.  The cost of secure computation at operator $o_i$ is represented by a function $cost_{o_i}({\bf N}_i)$, where ${\bf N}_i$ be the set of the cardinalities of the input tables.  The I$/$O cost of handling each intermediate result of operator $o_i$ by \sysname is mainly the cost to sort the exhaustively padded array of size $n'$ and the cost of copying the array of size $n$ to the new array of size $n'$, denote by $cost_{sort}(n)$ and $cost_{copy}(n,n')$ respectively. In particular, the size of the exhaustively padded array also depends on the sizes of the input tables, represented by a function $o_i({\bf N}_i)$. Most of the time, the size of the exhaustively padded array is the product of the input table sizes. For example, if operator $o_i$ takes one input table of size $n_1$, then $o_i({\bf N}=\{n_1\}) =n_1$; if operator $o_i$ takes two input tables of size $n_1$ and $n_2$, then $o_i({\bf N}=\{n_1,n_2\}) =n_1\cdot n_2$.  The number of tuples $n'$ copied to the new array depends on the noisy output size $n_i$.

The input tables to an operator can be a mix of base tables and the intermediate output tables. Given a database of $k$ relations $(R_1,\ldots, R_k)$, there will be an additional $l$ number of intermediate output tables generated in a query tree of $l$ operators. The public knowledge $K$ contains the table sizes for the original tables, but the output cardinalities of the operators are unknown. In the private setting, we cannot use the true cardinality for the estimation. Instead, \sysname uses the naive reduction factors from~\cite{selinger1979access} to estimate the output cardinality of each operator, denoted by $estimate(\carquery_i,K)$.  The noise $\eta_i$ added to the cardinality depends on the noise distribution of the truncated Laplace mechanism and we use the expectation of the distribution $\tlap(\epsilon_i,\delta_i,\Delta \carquery_i)$ to model the noise. Putting this together, the output cardinality $n_i$ at operator $o_i$ is estimated by $estimate(\carquery_i,K) + \mathbb{E}(\tlap(\epsilon_i,\delta_i,\Delta \carquery_i))$.

Based on this formulation, the cost of assigning privacy budget $P = \{(\epsilon_i,\delta_i) | i=1,\ldots,l\}$ to operators $Q=\{o_1,\ldots,o_l\}$ in the query tree is modeled as
\begin{equation}
C(P,K)  = \sum_{i=1}^l cost_{o_i}({\bf N}_i) + cost_{sort}(o_i({\bf N}_i)) + cost_{copy}(o_i({\bf N}_i),n_i)
\end{equation}
where $n_i= estimate(\carquery_i,K) + \mathbb{E}(\tlap(\epsilon_i,\delta_i,\Delta q_i))$.

\stitle{Optimization:} Using this cost model, we find the optimal solution to the following problem as the budget allocation strategy for Algorithm~\ref{algo:sw}.
\begin{eqnarray}
& \min_{P} C(P) ~~s.t. & \sum_{i=1}^{l} \epsilon_i \leq \epsilonper,  \sum_{i=1}^{l} \delta_i \leq \deltaper, \nonumber \\
&&\epsilon_i,\delta_i  \geq 0~~ \forall i=1,\ldots l
\end{eqnarray}
We show the detailed cost model including the unit cost function $cost_{o_i}({\bf N}_i, n_i)$, $cost_{sort}(o_i({\bf N}_i))$, and $cost_{copy}(o_i({\bf N}_i),n_i)$ for two different \mpc protocols in the next section. Using a convex optimization solver, we determine the optimal budget split that minimizes the cost objective function. We evaluate how it performs against the baseline approaches in Section~\ref{sec:results}. 

\change{
\section{Protocol Implementation}
\label{sec:implementation}
\sysname is a flexible computation layer that supports a wide range of {\federation}s and database engines. The only requirements are that the \federation process queries as operator trees and that the underlying database engine supports SQL queries.  \sysname implements a wide range of SQL operators for PDF queries, including selection, projection, aggregation, limit, and some window aggregates.  For joins, we handle equi-joins and cross products.  At this time we do not support outer joins or set operations. For output policy 2 where client sees noisy answers released from differentially private mechanisms, we support aggregate queries like {\tt COUNT} for the final data operator.

The fundamental design pattern of \mpc protocols is the circuit model, where functions are represented as circuits, i.e. a collection of topologically ordered gates. The benefit of this model is that by describing how to securely compute a single gate, we can compose the computation across all gates to evaluate any circuit, simplifying the design of the protocol. The downside is that circuits are difficult for programmers to reason about. Instead, the von Neumann-style Random Access Machine (RAM) model handles the impedance mismatch between programmers and circuits by creating data structures, such as ORAM, that utilize circuit-based oblivious algorithms for I$/$O. With the RAM model, programmers no longer have to write programs as circuits. Instead, they can think of their data as residing in secure arrays that can be plugged into their existing code. The downside of the RAM model is that since the data structures are more general purpose, performance lags behind direct circuit model implementations of the same function. In this work, we use \sysname with both RAM model and circuit model protocols, providing a general framework that can easily add new protocols to the query executor.


\subsection{RAM Model}
When integrating a new protocol, we first create an operator cost model that captures the execution cost of each operator. For a RAM based protocol, we model the cost of secure computations $cost_{o}({\bf N})$ for commonly used relational operators in Table~\ref{tbl:io-cost}, where ${\bf N}$ is the set of cardinalities of input tables of operator $o$. The cost function considered in this work focuses on the I$/$O cost of each relational operator in terms of reads and writes from a secure array.

We let $c_{read}(n)$ and $c_{write}(n)$ be the unit cost for reading and writing of a tuple from a secure array of size $n$ respectively. Secure arrays protect their contents from attackers by guaranteeing that all reads or writes access the same number of memory locations. Typically, the data structures will shuffle either their entire contents or their access path on each access, with different implementations providing performance ranging from $O(\log n)$ to $O(n\log (n)^2)$ per read or write~\cite{Arasu2014,Wang2014}. This variable I$/$O cost guarantees that an observer cannot look at memory access time or program traces to infer sensitive information.

Given the unit cost for reading and writing, we can represent the cost of secure computation for each operator as the sum of reading cost and writing cost. For example, the cost function for a Filter operator $cost_{Filter}({\bf N}=\{n_1\})$ is the cost of reading $n_1$ number of input records and writing $n_1$ number of output records. For a Join operator over two input tables, the cost function $cost_{Join}({\bf N}=\{n_1,n_2\})$ equals to the sum of $n_1$ number of records reading from the first secure array of size $n_1$, $(n_1*n_2)$ number of record reading from the second array of size $n_2$, and  $(n_1*n_2)$ number of record writes, i.e. $n_1* c_{read}(n_1) + n_1*n_2* c_{read}(n_2) + n_1*n_2* c_{write}(n_1* n_2)$. The other operators are modeled similarly in Table~\ref{tbl:io-cost}. The cost function for copying from an array of size $n$ to an array of size $n'$ can be modeled as $cost_{copy}(n,n') = n' * c_{read}(n) + n'*c_{write}(n')$.

In this work, we use ObliVM~\cite{Liu2014} to implement our RAM based relational operators, under a SMCQL~\cite{Bater2017} \federation design. We chose ObliVM due to its compiler, which allows SMCQL to easily convert operators written as C-style code into \mpc programs. However, \sysname can support more recent ORAM implementations such as DORAM~\cite{Doerner2017} by summarizing their cost behavior as in Table~\ref{tbl:io-cost} and using those operator costs in our cost model.  In SMCQL, private intermediate results are held in secure arrays, i.e., oblivious RAM (ORAM), and operators are compiled as circuits that receive an ORAM as input, evaluate the operator, and output the result ORAM to the next operator. With \sysname, we resize the intermediate ORAM arrays passed between operators by obliviously eliminating excess dummy tuples, reducing the intermediate cardinalities of our results and improving performance. Any ORAM implementation is compatible with \sysname, as long as its I$/$O characteristics can be similarly defined as in Table~\ref{tbl:io-cost}.

\begin{table}
\centering
{\small
\begin{tabular}{|l|c|}
\hline
  Operator & $cost_{o}({\bf N}=\{n_1,n_2,\ldots\})$\\
\hline
{\it Filter} & $n_1 * c_{read}(n_1) + n_1 * c_{write}(n_1)$ \\
{\it Join} &  $n_1 * c_{read}(n_1) + n_1 * n_2 *c_{read}(n_2)$\\
& $+ n_1 * n_2 * c_{write}(n_1 * n_2)$ \\
{\it Aggregate} & $n_1 * c_{read}(n_1) + c_{write}(n_1)$ \\
{\it Aggregate, Group By} & $n_1 * c_{read}(n_1) + n_1 * c_{write}(n_1)$ \\
{\it Window Aggregate} & $n_1 * c_{read}(n_1) + n_1 * c_{write}(n_1)$ \\
{\it Distinct} & $n_1 * c_{read}(n_1) + n_1 * c_{write}(n_1)$ \\
{\it Sort} & $n_1 * \log^2 (n_1) * (c_{read}(n_1) + c_{write}(n_1))$\\
\hline
\end{tabular}
}
\caption{Operator I$/$O cost} \label{tbl:io-cost}
\end{table}

\subsection{Circuit Model}

Alternatively, if a \federation uses a circuit based protocol to express database operators, we can estimate the cost of an operator $o$, $cost_{o}({\bf N})$,  with $n_{gates}$ gates in its circuit, a maximum circuit depth of $d_{circuit}$, $n_{in}$ elements of input, and an output size of $n_{out}$ as: $$cost_{op}({\bf N}) = c_{in}*n_{in} + c_g*n_{gates} + c_d * d_{circuit} + c_{out} * n_{out},$$
where $c_{in}$ and $c_{out}$ account for the encoding and decoding costs needed for a given protocol,  $c_{g}$ and $c_{d}$ account for the costs due to the gate count and circuit depth respectively, the number of input element is $n_{in}=\sum_{n_j \in {\bf N}} n_j$ and the number of output element is $n_{out} = \prod_{n_j \in {\bf N}} n_j$.


The generic nature of this model allows us to capture the performance features of all circuit-based \mpc protocols. For example, in this work we utilize EMP toolkit~\cite{emp-toolkit}, a state of the art circuit protocol, in conjunction with \sysname.  As EMP has not been used in {\federation}s, we implemented EMP with SMCQL. Instead of compiling operators into individual circuits, EMP compiles a SQL query into a single circuit, though it still pads intermediate results within the circuit. In this setting, we can still apply \sysname as we did in the RAM model, where we minimize the padding according to our differential privacy guarantees. Instead of modifying the ORAM, we directly modify the circuit to eliminate dummy tuples.

For both our RAM model and our circuit model implementations, \sysname uses a cost model to estimate the execution cost of intermediate result padding. With the model, we can carry out our privacy budget optimization and allocation as seen in Section~\ref{sec:optimization}. {\sysname}'s protocol agnostic design allows it to optimize the performance of \mpc, regardless of the implementation.
\begin{table*}
\scriptsize
\centering
\begin{tabular}{|l|l|}
\hline
{\bf Name}  & Query \\
\hline
	{\it Dosage Study} & {\tt\scriptsize SELECT DISTINCT d.pid FROM diagnoses d, medications m WHERE d.pid = m.pid AND medication = 'aspirin' AND}\\
	& {icd9 = 'circulatory disorder' AND dosage = '325mg'}\\
	\hline
{\it Comorbidity} & {\tt\scriptsize SELECT diag, COUNT(*) cnt FROM diagnoses WHERE pid $\in$ cdiff\_cohort $\land$ diag $<>$ 'cdiff' ORDER BY cnt  DESC LIMIT 10;} \\
\hline
{\it Aspirin} & {\tt\scriptsize SELECT COUNT(DISTINCT pid) FROM diagnosis d JOIN medication m on d.pid = m.pid}\\
{\it Count} &   {\tt\scriptsize JOIN demographics demo on d.pid = demo.pid WHERE d.diag = 'heart disease' AND m.med = 'aspirin' AND d.time <= m.time; } \\
\hline
{\it 3-Join} & {\tt\scriptsize SELECT COUNT(DISTINCT pid) FROM diagnosis d JOIN medication m on d.pid = m.pid}\\
 &   {\tt\scriptsize JOIN demographics demo on d.pid = demo.pid JOIN demographics demo2 ON d.pid = demo2.pid}\\
 &  {\tt\scriptsize WHERE d.diag = 'heart disease' AND m.med = 'aspirin' AND d.time <= m.time; } \\
\hline
\end{tabular}
\caption{HealthLNK query workload.}
\label{tbl:operator-summary}
\vspace{-2.9mm}
\end{table*}

\subsection{M-Party Support} \label{sec:m-party}
In this work, we implement \sysname using two party \mpc protocols, meaning that we run our experiments over two data owners. However, \sysname supports additional data owners. By leveraging the associativity of our operators, we can convert all $m$ party computations into a series of binary computations. For example, assume we want to join to tables $R$ and $S$, which are horizontally partitioned across 3 data owners as $R1, R2, R3$, and $S1, S2, S3$. We can join across the 3 data owners as: $(R1 \bowtie S1) \cup (R1 \bowtie S2) \cup (R3 \bowtie S3) \cup (R2 \bowtie S1) \cup (R2 \bowtie S2) \cup (R2 \bowtie S3) \cup (R3 \bowtie S1) \cup (R3 \bowtie S2) \cup (R3 \bowtie S3)$. We can execute this series of binary joins using two party secure computation then union the results. With this construction, we can scale \sysname out to any number of parties. However, this algorithm is not efficient due to naively carrying out $m^2$ secure computation operations for every $m$ parties we compute over. We can use $m$ party \mpc protocols~\cite{Asharov2013,Beaver1990,Goldreich1987}, which have fewer high-level operators and made large performance strides, but these protocols are still expensive. We leave performance improvements in this setting for future work.
}



\section{Results}
\label{sec:results}

\change{We now look at \sysname performance using both real world and synthetic datasets and workloads. First, we cover our experimental design and setup. Next, we evaluate the end to end performance of our system and show the performance vs privacy trade-off under both true and differentially private answer settings. Then, we look at how \sysname performs at the operator level and the accuracy of our cost model, followed by a discussion of our budget splitting strategies. Finally, we examine how \sysname scales to larger datasets and more complex queries.}

\subsection{Experimental Setup}
For this work, we implemented \sysname on top of an existing healthcare database federation that uses SMCQL~\cite{Bater2017} for \mpc. This \federation serves a group of hospitals that wish to mine the union of their electronic health record systems for research while keeping individual tuples private. 

A clinical data research network or CDRN is a consortium of healthcare sites that agree to share their data for research. CDRN data providers wish to keep their data private. We examine this work in the context of HealthLNK~\cite{healthlnkirb}, a CDRN prototype for Chicago-area healthcare sites. This repository contains records from seven Chicago-area healthcare institutions, each with their own member hospitals, from 2006 to 2012, totaling about 6 million records. The data set is selected from a diverse set of hospitals, including academic medical centers, large county hospitals, and local community health centers.

Medical researchers working on HealthLNK develop SQL queries and run them on data from a set of healthcare sites. These researchers are the clients of the \federation and the healthcare sites are the private data owners. 

\stitle{Dataset:} We evaluate \sysname on medical data from two Chicago area hospitals in the HealthLNK data repository~\cite{healthlnkirb} over one year of data. This dataset has 500,000 patient records, or 15 GB of data. \change{For additional evaluation, we generate synthetic data up to 750 GB based on this source medical data.} To simplify our experiments, we use a public patient registry for common diseases that maintains a list of anonymized patient identifiers associated with these conditions. We filter our query inputs using this registry.

\stitle{Query Workload:}
For our experiments, we chose three representative queries based on clinical data research protocols~\cite{hernandez2015adaptable, cdiff2015irb} and evaluate \sysname on this workload using de-identified medical records from the HealthLNK data repository.  We also generate synthetic versions of {\it Aspirin Count} that contain additional join operators. We refer to these queries by the number of join operators, e.g., 3-Join for 3 join operators. The queries are shown in Table~\ref{tbl:operator-summary}.

\stitle{Configuration:} \sysname query processing runs on top of PostgreSQL 9.5 running on Ubuntu Linux. We evaluated our two-party prototype on 6 servers running in pairs. The servers each have 64 GB of memory, 7200 RPM NL-SAS hard drives, and are on a dedicated 10Gb/s network. Our results report the average runtime of three runs per experiment. We implement \sysname under both the RAM model (using ObliVM~\cite{Liu2015}) and circuit model (using EMP\cite{emp-toolkit}). For most experiments, we show results using the RAM model, though corresponding circuit model results are similar. Unless otherwise specified, the results show the end-to-end runtime of a query with output policy 1, i.e., true results.


\subsection{End-to-end Performance}
In our first experiment, we look at the the end to end performance of \sysname compared to baseline {\federation} execution. For execution, we set $\epsilon$ to 0.5, $\delta$ to 0.00005, and use the optimal budget splitting approach. 

\begin{figure}[!t]
\centering
  \begin{subfigure}[b]{0.8\columnwidth}
  \centering
    \includegraphics[width=\columnwidth]{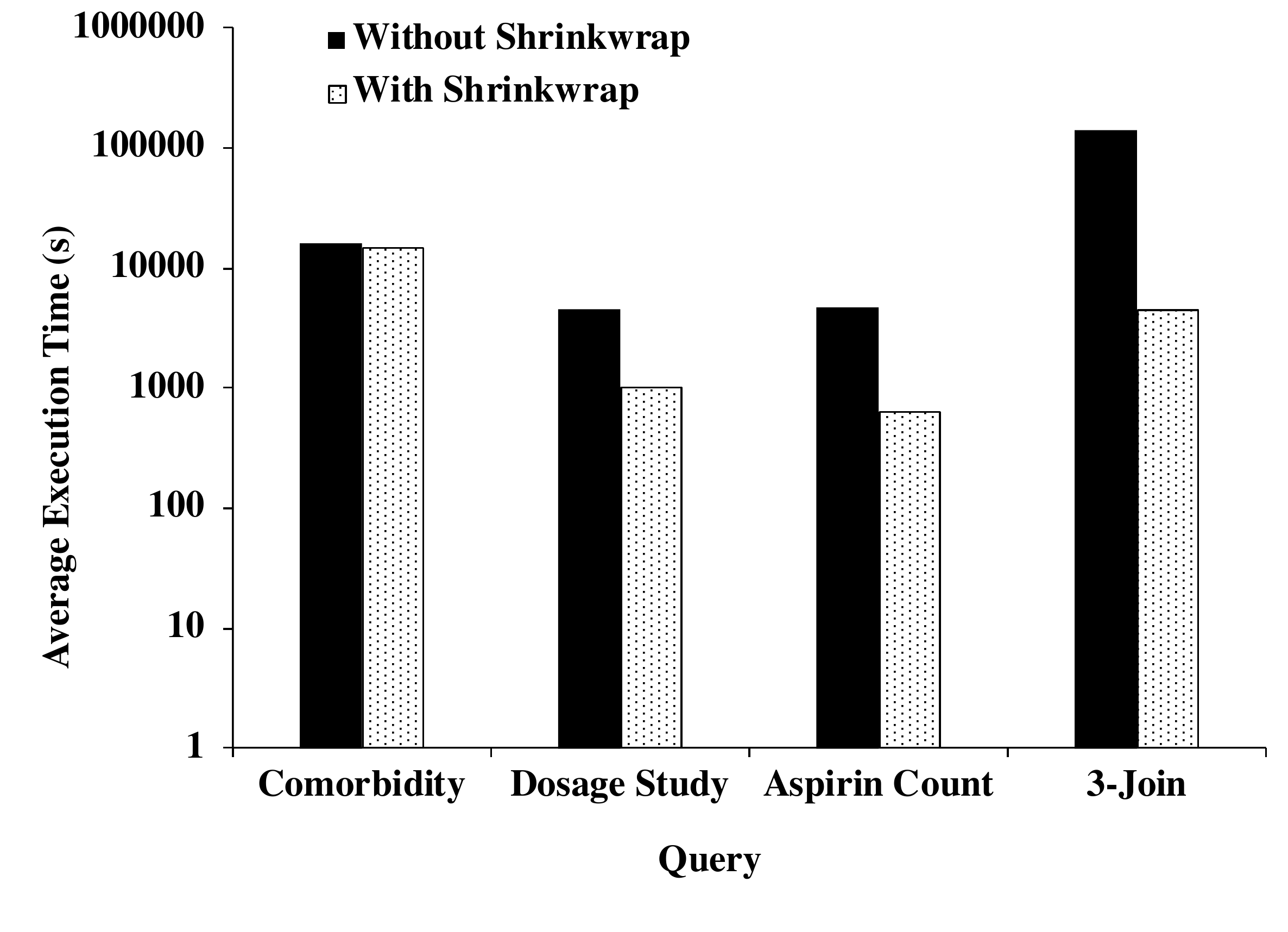}
    \caption{RAM Model (ObliVM), }
    \label{fig:end-to-end-ram}
  \end{subfigure}
  \begin{subfigure}[b]{0.8\columnwidth}
  \centering
    \includegraphics[width=\columnwidth]{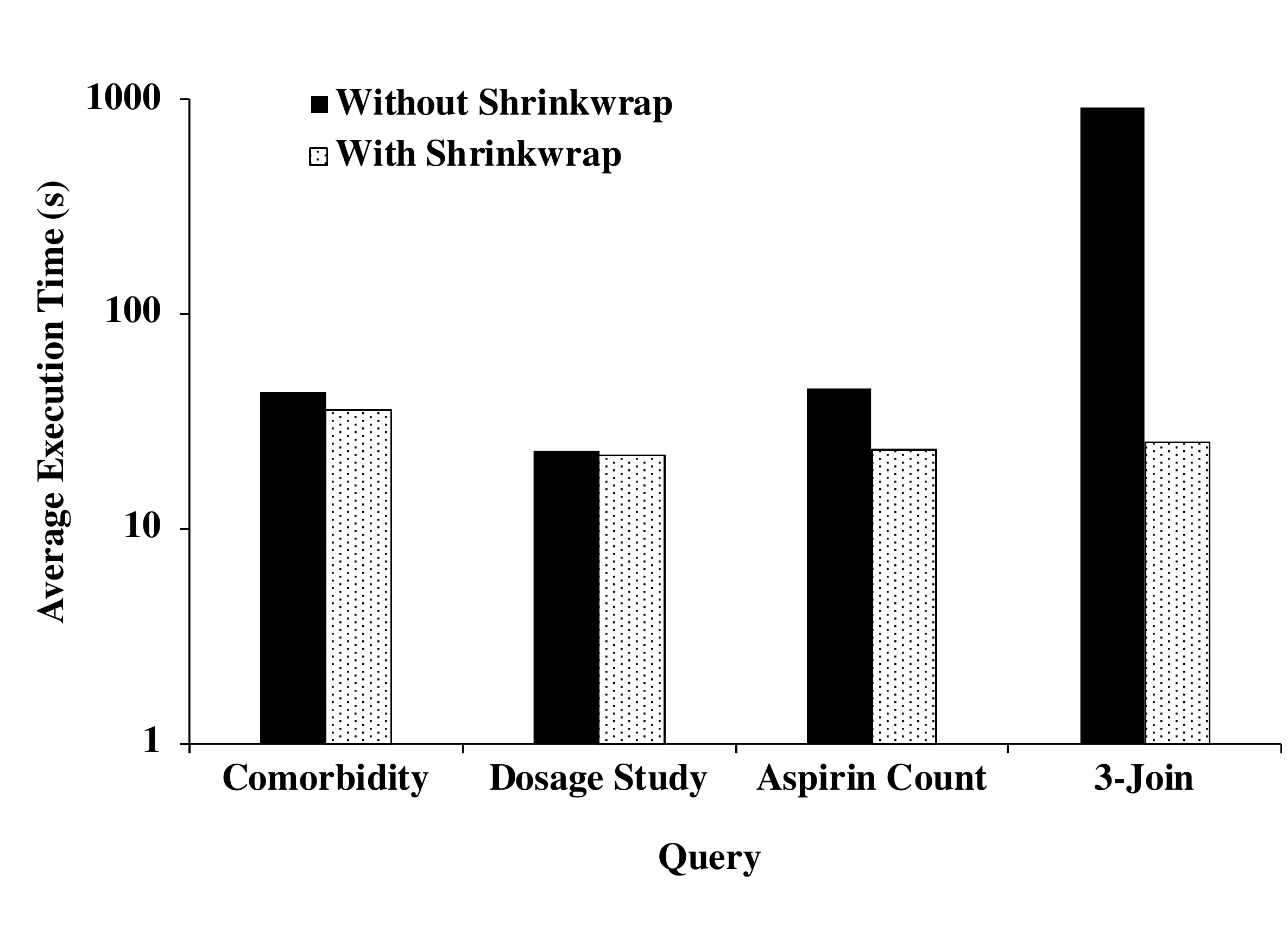}
    \caption{Circuit Model (EMP)}
    \label{fig:end-to-end-emp}
  \end{subfigure}
  \caption{End-to-End Shrinkwrap Performance for All Queries at $\epsilon$ = 0.5, $\delta$ = .00005.}
  \label{fig:end-to-end}
\end{figure}

\change{In Figure~\ref{fig:end-to-end}, we look at the overall performance of four queries under \sysname for both RAM model and circuit model \mpc. For {\it Comorbidity}, the execution does not contain any join operators, meaning that there is no explosion in intermediate result cardinality, so \sysname provides fewer benefits. {\it Dosage Study} contains a join operator with a parent distinct operator. Applying differential privacy to the output of the join improves performance by close to 5x under the RAM model, but sees fewer benefits under the circuit model. {\it Aspirin Count} contains two joins and sees an order of magnitude improvement under the RAM model and a 2x improvement under the circuit model. Finally, {\it 3-Join} has an enormous cardinality blow-up, allowing \sysname to improve its performance by 33x under the RAM model and 35x under the circuit model. In Section~\ref{sec:join-scaling}, we examine this effect of join operators on performance. Although implementation in the circuit model outperforms the RAM model by an order of magnitude, Shrinkwrap gives similar efficiency improvements in both systems.}


\subsection{Privacy, Performance, and Accuracy}

We consider the performance vs privacy trade-off provided by \sysname and evaluate the effect of the execution privacy budget on query performance. For this experiment, we use the {\it 3-Join} query with the optimal budget splitting approach, fix the total privacy budget $\epsilon=1.5$ and $\delta=0.00005$, and use output setting 2, i.e. differentially private final answers. We vary the privacy budget for performance $\delta_{1-l}$ from 0.1 to 1.5  and fix $\deltaper = \delta$. The remaining privacy budget $\epsilon-\epsilonper$ is spent on the query output using Laplace mechanism.

\begin{figure}[!t]
\centering
  \begin{subfigure}[b]{0.8\columnwidth}
  \centering
\includegraphics[width=\columnwidth]{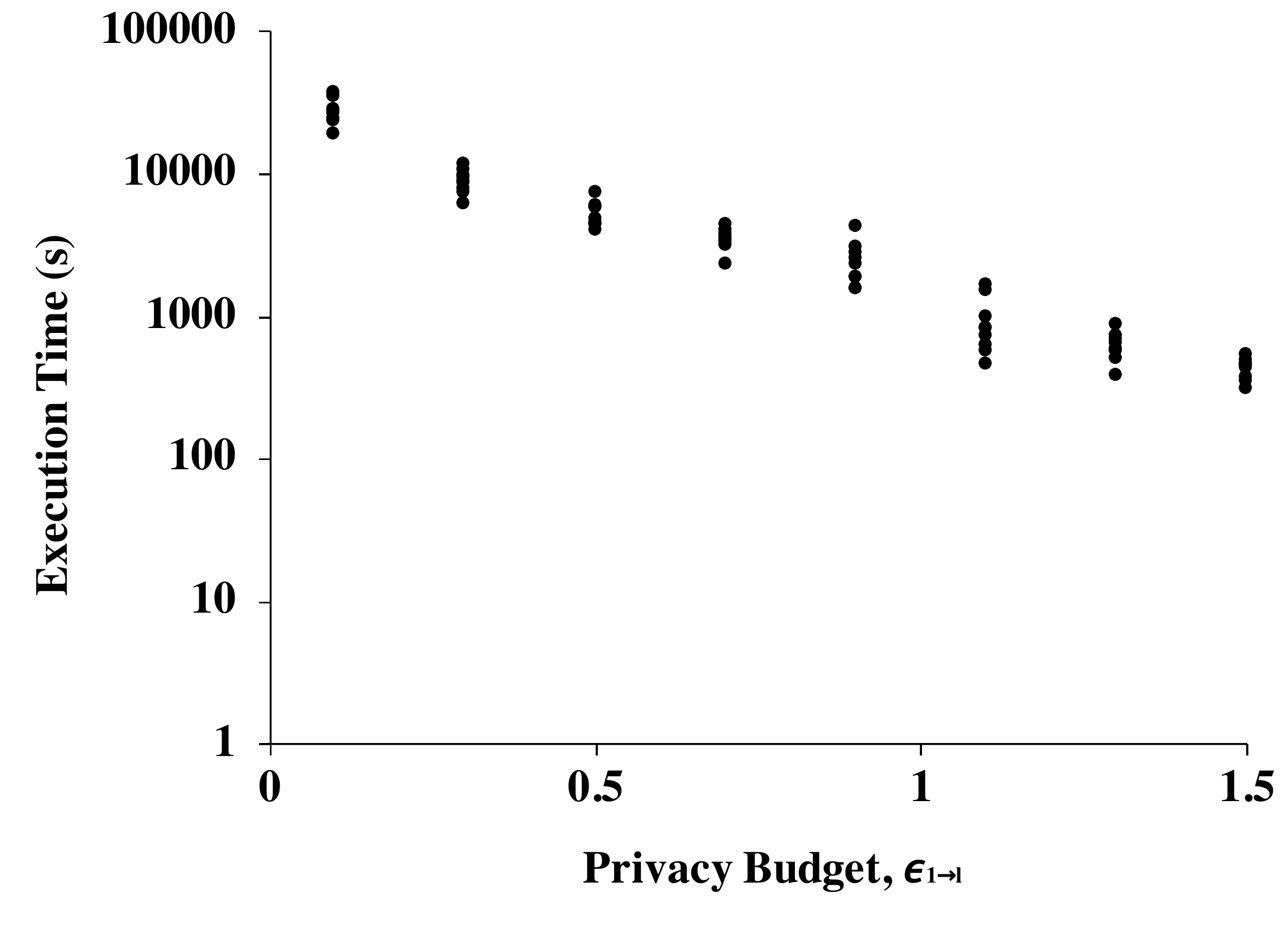}
\caption{Performance v.s. $\epsilonper$.}
\label{fig:tradeoff}
  \end{subfigure}
\begin{subfigure}[b]{0.8\columnwidth}
 \centering
\includegraphics[width=\columnwidth]{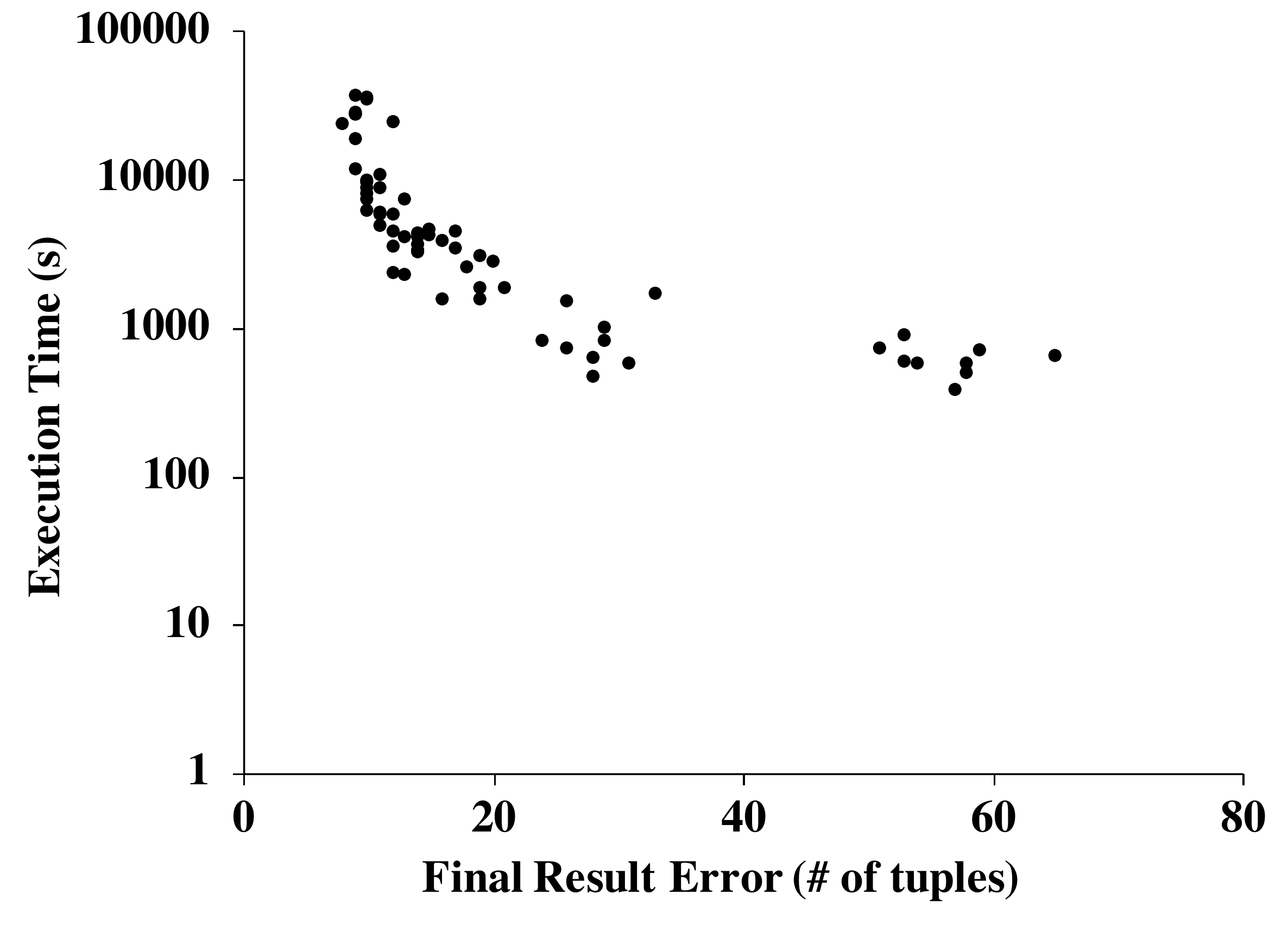}
\caption{Performance v.s. Error.}
\label{fig:accuracy}
  \end{subfigure}
  \caption{Privacy, Performance, and Accuracy Tradeoffs. Computed using different levels of privacy budget for performance $\epsilonper\in \{0.1,...,1.5\}$. Executed using RAM Model, {\it 3-Join}, $\epsilon=1.5, \delta= 0.00005$.}
\end{figure}

In Figure~\ref{fig:tradeoff}, we see the execution time as a function of the privacy budget for performance $(\epsilonper)$. As the privacy budget for performance $\epsilonper$ increases, the execution time decreases. A lower budget means that more noise must be present in the intermediate result cardinalities, which translates to larger cardinalities and higher execution time variance. We know that I$/$O cost dominates the \sysname cost model, so larger cardinalities mean more I$/$O accesses, which causes lower performance.


\change{
  In Figure~\ref{fig:accuracy}, we show the accuracy versus privacy trade-off by varying the privacy budget for performance $\epsilonper$. We report the execution time and the error introduced to the final query answer at various privacy budget for performance $\epsilonper< 1.5$. Out of a total possible output size of 5500 tuples, the query output is $10/5500=0.18\%$ of the total possible output size. If we allow additive noise of $0.36\%$ ($20/5500$), we get a 100x performance improvement. Since the total privacy $\epsilon$ is fixed, the more we spend for performance (larger $\epsilonper$), the less privacy budget we have at the output, i.e., as our execution time improves, our output noise increases. 

}


\subsection{Evaluating Budget Splitting Strategies} \label{sec:budget-splitting}
We first examine the relative performance of the three budget splitting strategies introduced in Section~\ref{sec:implementation}: uniform, eager, and optimal.  Recall that uniform splits the budget evenly across all the operators, eager inserts the entire budget at the first operator, and uses the \sysname cost model, along with cardinality estimates, to identify an optimal budget split.  \change{We also include an \textit{oracle} approach that shows the performance of \sysname if true cardinalities are used in the cost model to split the budget instead of cardinality estimates. Note that \textit{oracle} does not satisfy differential privacy, but gives an upper bound on the best performance achievable through privacy budgeting.} For this experiment, we use the {\it Aspirin Count} and {\it 3-Join} queries since they contain multiple operators where \sysname generates differentially-private cardinalities and set the privacy parameters usable during execution to $\epsilon$ = 0.5 and $\delta$ = .00005.

\begin{figure} [t]
\centering
\includegraphics[width=0.8\columnwidth]{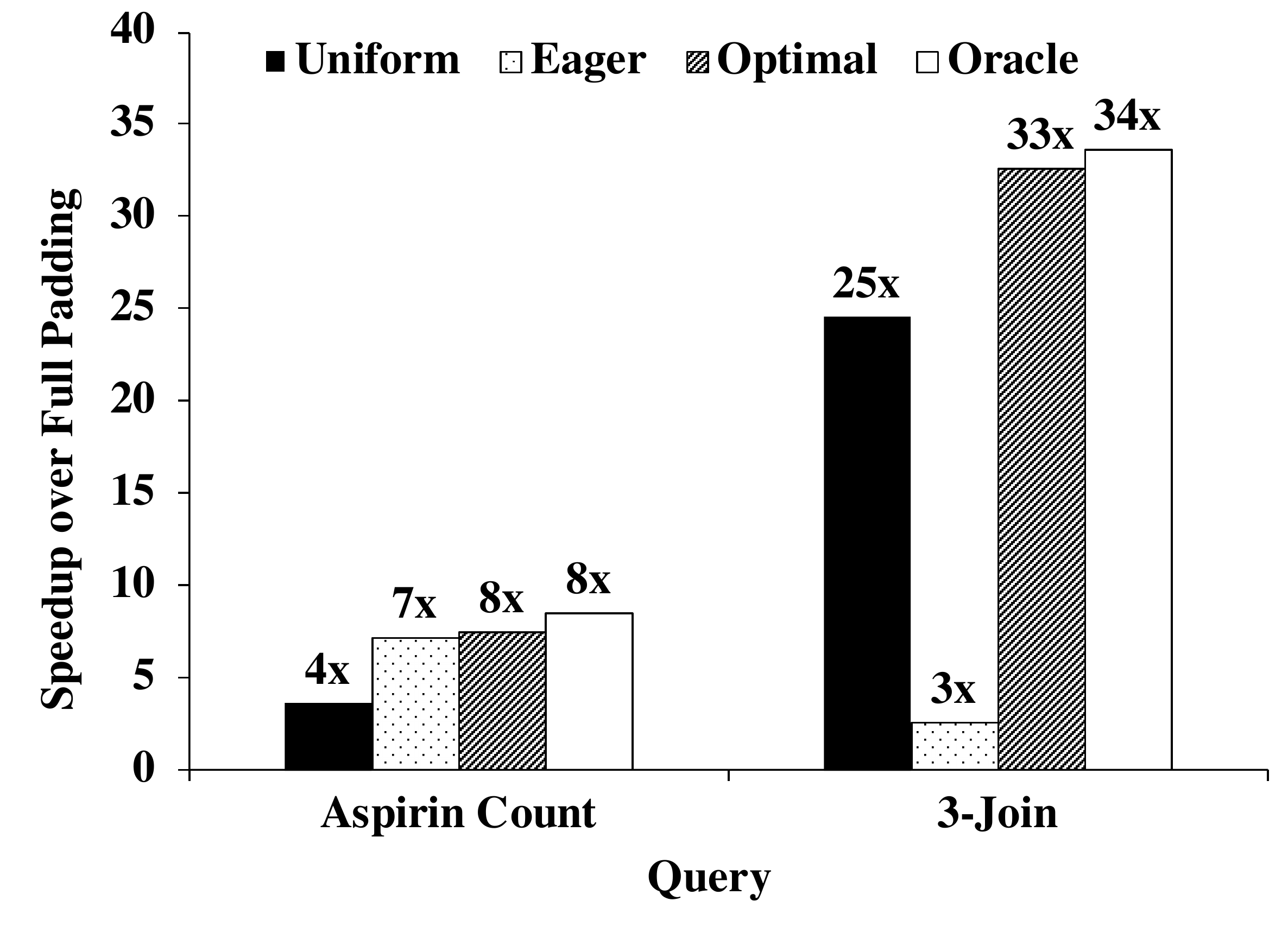}
\caption{Shrinkwrap execution speedup over baseline using different budget strategies. Executed using RAM model, $\epsilon$ = 0.5, $\delta$ = .00005}
\label{fig:speedup}
\end{figure}

Figure~\ref{fig:speedup} displays the relative speedup of all three approaches and the oracle for both queries over the baseline, fully padded, {\federation} execution. All three \sysname approaches provide significant performance improvements, ranging from 3x to 33x. As expected, optimal performs best for both queries. We also see that eager performs better for {\it Aspirin Count}, while uniform performs better for {\it 3-Join}. 

The benefits of minimized intermediate result cardinalities cascades as those results flow through the operator tree. By using the entire budget on the first join operator in {\it Aspirin Count}, the eager approach maximizes the effect of the intermediate result cardinality reduction, and as a result, outperforms the uniform approach. However, for the {\it 3-Join} query, the presence of the additional join operator overrides the cascading cardinality benefit of the first join operator. The uniform approach outperforms the eager approach by ensuring that all three of the join operators receive differentially private cardinalities. The optimal approach outperforms both eager and uniform by combining the best of both worlds. Optimal applies differential privacy to all operators, like uniform, but it uses a larger fraction of the budget on earlier operators, like eager.        

\change{Looking at the oracle approach, we can evaluate the accuracy of our cost model. For privacy reasons, our cost model does not use the true cardinalities. Here, we see that the optimal approach, which uses estimated cardinalities, and the oracle approach, which uses true cardinalities, provide similar performance. In experiments using \textit{Aspirin Count}, we calculate the correlation coefficient between the true execution time and the estimated execution time based on our cost model as .998 for the circuit model and .931 for the RAM model, given a $\epsilon$ = 0.5, $\delta$ = .00005. We see that the \sysname cost model provides a reasonably accurate prediction of the true cost.}

\subsection{Operator Breakdown}

Now, we look at the execution time by operator to see where \sysname provides the largest impact. We include only private operators in the figure. For this experiment we use the {\it Aspirin Count} query with $\epsilon$ = 0.5 and $\delta$ = 0.00005. We show the execution times for the baseline, fully padded approach for intermediate results, as well as the uniform, eager, and optimal budget splitting approaches. 

\begin{figure} [t]
\centering
\includegraphics[width=0.8\columnwidth]{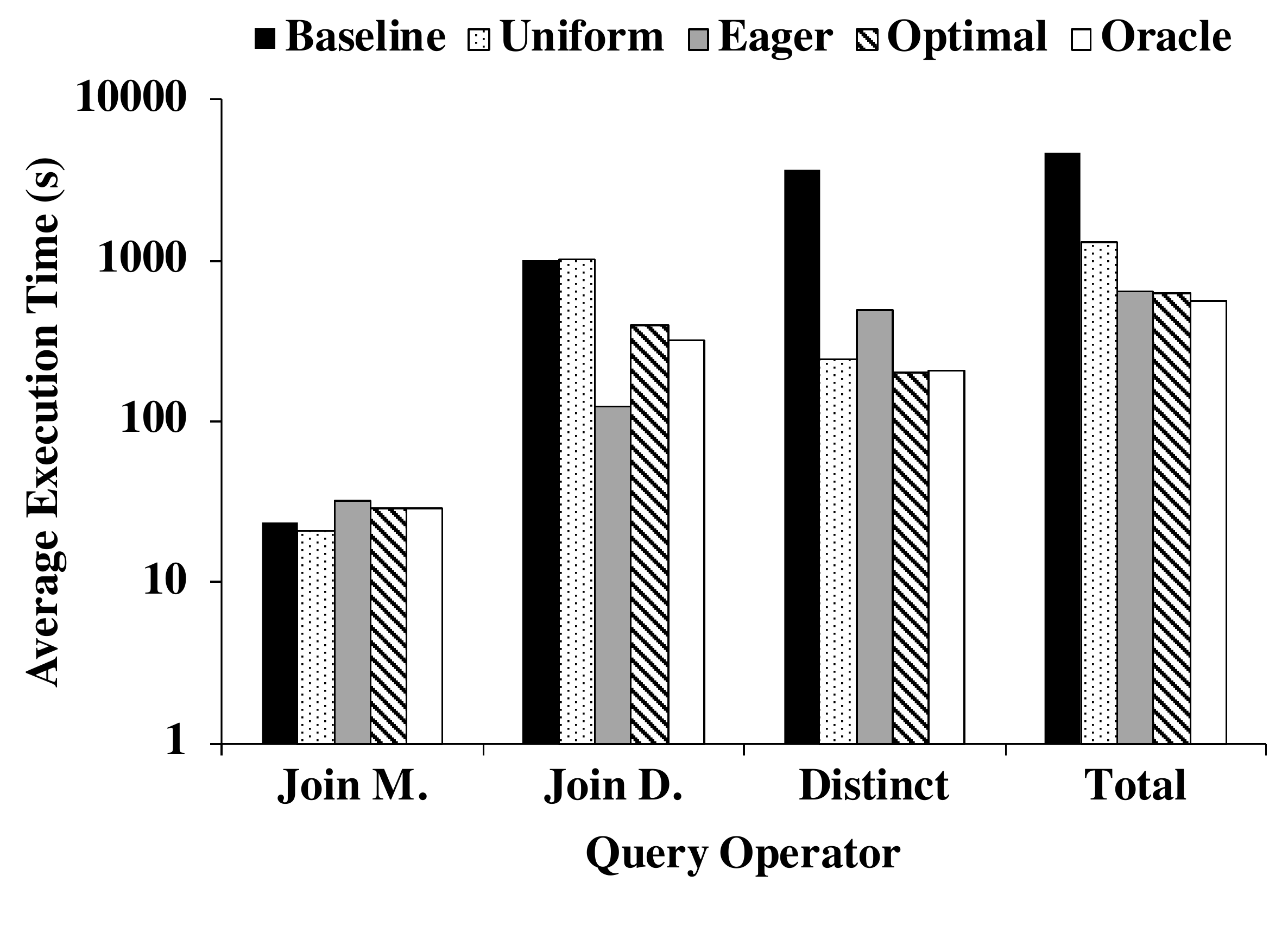}
\caption{Per operator performance using different budget strategies. Executed using \textit{Aspirin Count}, RAM model, $\epsilon$ = 0.5, $\delta$ = .00005}
\label{fig:aspirin-operator}
\end{figure}

Figure~\ref{fig:aspirin-operator} shows the execution times of each operator in the {\it Aspirin Count} query. Note that \sysname does not apply any differential privacy until the output of Join M. (Join on Medications), so we see no speed-up in the actual Join M. operator, only in the later query operators. 

The first operator to see a benefit from \sysname is Join D. (Join on Demographics), the second join in the query. \sysname generates a differentially private cardinality for the output of the previous join, Join M., according to the selected budget splitting approach. For uniform, the allocated budget is not large enough to generate a differentially-private cardinality that is smaller than the fully padded cardinality. In fact, the overhead of calculating the differentially-private cardinality actually causes the operator runtime to go above the baseline runtime. For eager, \sysname spends the entire budget on the output of the first join and reduces the intermediate cardinality by an order of magnitude. The optimal approach uses enough budget to reduce the intermediate cardinality, but not as much as eager.   

The next operator, DISTINCT, also sees a large performance benefit. All three budget splitting approaches reduce the runtime for DISTINCT by about an order of magnitude. Since the baseline intermediate cardinality for DISTINCT is much larger, the uniform approach generates a differentially-private cardinality that reduces the intermediate result size and improves performance. The eager approach does not have any additional budget to use, but the effect of the previous intermediate cardinality cascades through the query tree and reduces the execution time for DISTINCT as well. Finally, the optimal approach applies its remaining budget to see a significant performance improvement.

In the {\it Aspirin Count} query, the three budget splitting approaches allocate the budget between three operators and provide an insight into the trade-off between early and late budget allocation. We see wildly varying execution times for each of the operators, with different operators providing the bulk of the performance cost depending on the budget strategy. Here, the optimal strategy gives the best performance. The substantial variance in execution for these operators demonstrates the value of the added accuracy that our I$/$O cost model provides. 

\subsection{Join Scale Up} \label{sec:join-scaling}
We now look at how \sysname scales with more complex workloads. From our experiments, we know that the number of joins in a query has an extremely large impact on execution time. \change{More complex workloads typically contain more nested statements and require more advanced SQL processing, but can be broken down into a series of simpler statements. This work focuses on complexity as a function of the join count in order to target the largest performance bottleneck in a single SQL statement.}

In our experiment, we scale the join count and measure the performance. Since \sysname applies differential privacy at the operator outputs, we measure the performance impact of join operators by looking at the execution time of their immediate parent operator. For consistency, we truncate the number of non-inherited input tuples of each join to equal magnitudes and use $\epsilon$ = 0.5 and $\delta$ = 0.00005 for all runs. 

From Figure~\ref{fig:join-scaling}, we see how execution time scales as a function of the join count. As the number of joins increases, the execution time also increases. To find the source of the performance slowdown, we can look at the \sysname cost model. For each join operator, the baseline {\federation} fully pads the output. If the join inputs were both of size $n$, then the output would be $n^2$. As the number of joins increase, the cardinality rises from $n$ to $n^2$ to $n^3$. The exponential increase in output cardinality requires significantly more I$/$O accesses, which cause the large slow down in performance. \sysname applies differential privacy to reduce the output size at each join, reducing the magnitude of each join output. While \sysname still sees a significant growth in execution time as a function of the join count, the overall performance is orders of magnitudes better.   

\begin{figure}
\centering
\includegraphics[width=0.8\columnwidth]{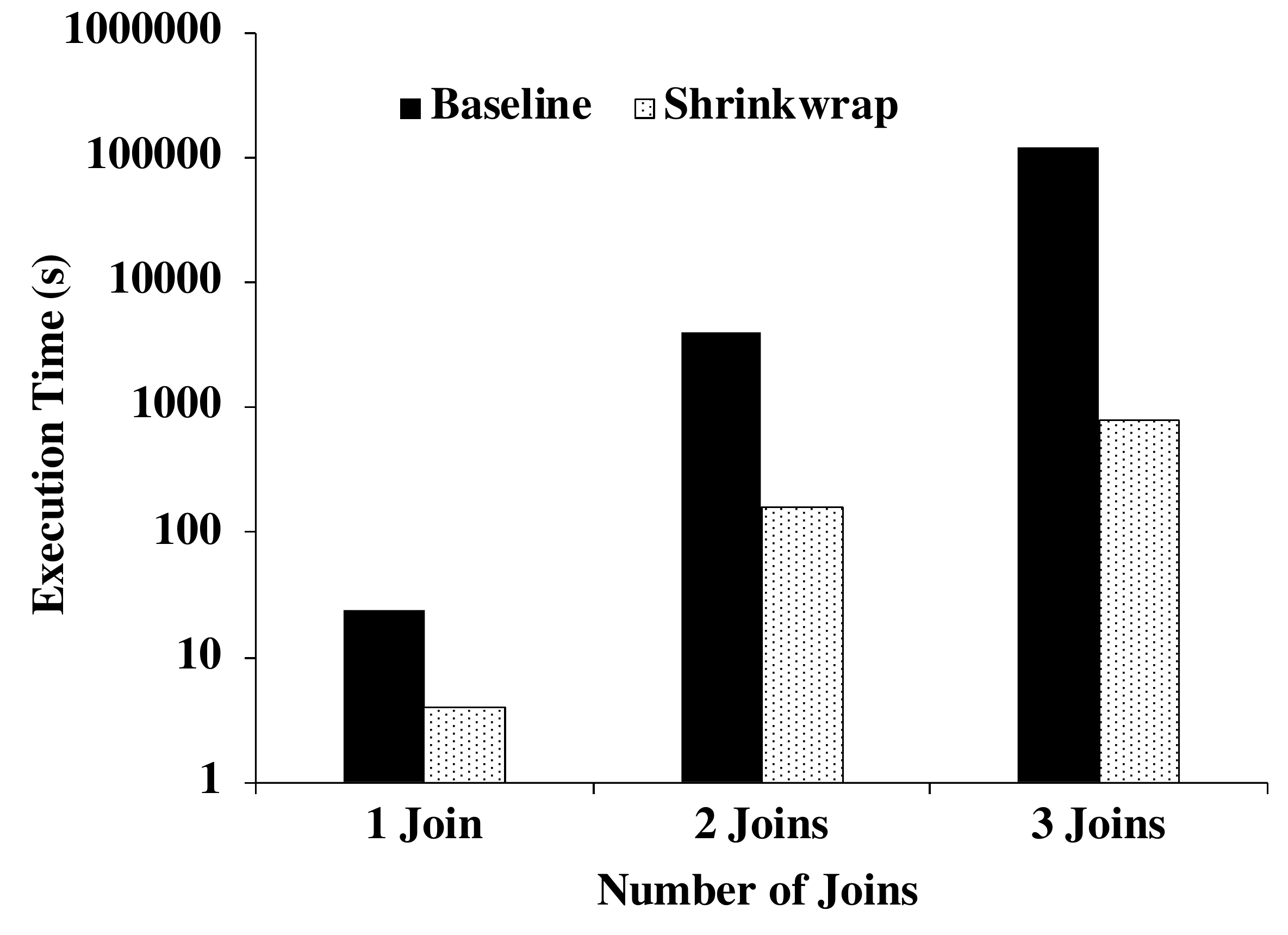}
\caption{\textit{Aspirin Count} with synthetic join scaling. Executed using RAM model. $\epsilon$ = 0.5, $\delta$ = .00005}
\label{fig:join-scaling}
\end{figure}

\eat{
\begin{figure}
\centering
  \begin{subfigure}[b]{0.45\textwidth}
  \centering
\includegraphics[width=0.8\textwidth]{figures/join-scaling}
\caption{\textit{Aspirin Count} with synthetic join scaling. Executed using RAM model. }
\label{fig:join-scaling}
  \end{subfigure}
  \begin{subfigure}[b]{0.45\textwidth}
  \centering
\includegraphics[width=0.8\textwidth]{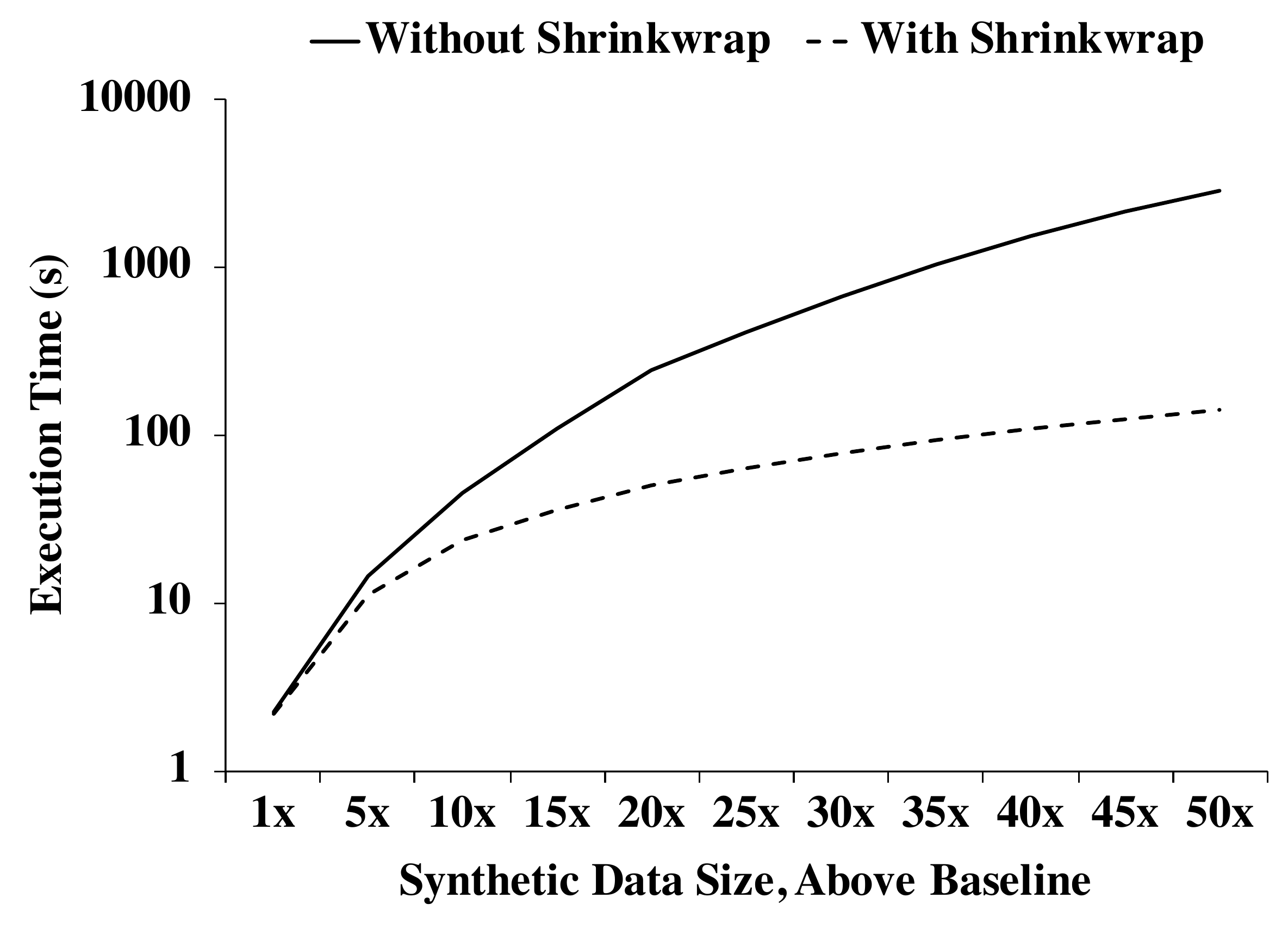}
\caption{\textit{Aspirin Count} with synthetic data scaling. Executed using Circuit model.}
\label{fig:emp}
  \end{subfigure}
  \caption{Join and Data Scaling with Shrinkwrap,  $\epsilon$ = 0.5, $\delta$ = .00005 }
  \label{fig:scale}
\end{figure}
}

\change{
\subsection{Data Size Scale Up} \label{sec:data-scaling}
We evaluate the performance of \sysname with increasing data sizes. We leverage the higher efficiency of our circuit model protocol, EMP~\cite{emp-toolkit}, to examine the effect of \sysname on larger input data sizes. Without \sysname, the RAM model protocol, ObliVM~\cite{Liu2015}, needs 1.3 hours to run the \textit{Aspirin Count} query and at 50x, 65 hours to complete. Instead, we use EMP which completes in 15 minutes.

\begin{figure}[!t]
\centering
\includegraphics[width=0.8\columnwidth]{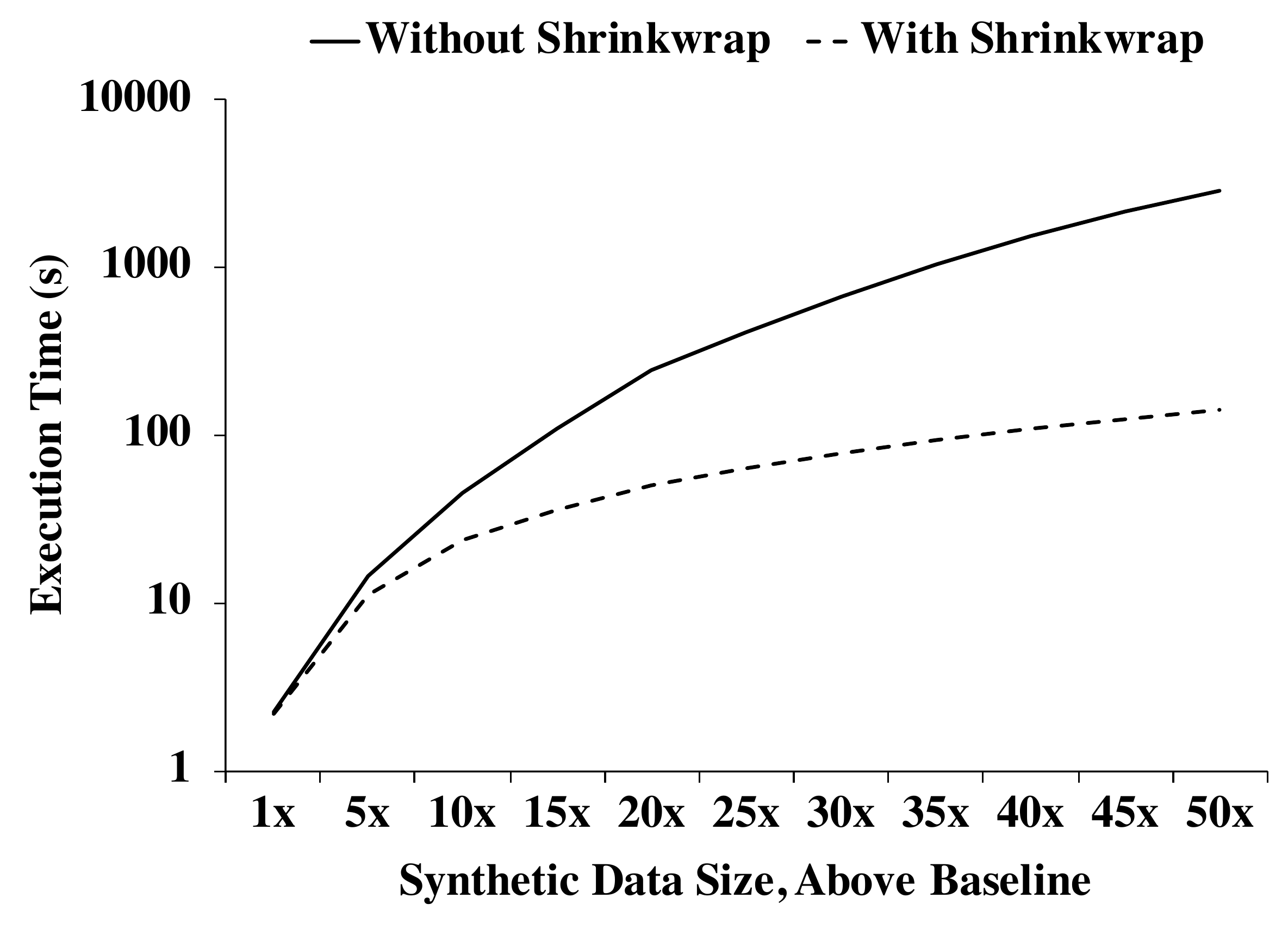}
\caption{\textit{Aspirin Count} with synthetic data scaling. Executed using Circuit model. $\epsilon$ = 0.5, $\delta$ = .00005} \label{fig:emp}
\end{figure}

 Figure~\ref{fig:emp} shows the execution time both with and without \sysname. We used a limited implementation of EMP~\cite{emp-toolkit} and our {\it Aspirin Count} reference query. We generated synthetic data that duplicates the original tables up to 50 times the original 15 GB data size, giving us a maximum effective data size of 750 GB. With Figure~\ref{fig:emp}, we see that \sysname still provides a significant improvement during execution. As the data size grows, our performance improvement grows as well, reflecting both the power and flexibility of \sysname. 
}


\section{Related Work}
\label{sec:related}

Within the literature, different approaches exist to improve query processing in {\federation}s, such as databases based on homomorphic encryption, TEEs, differential privacy, and cloud computation~\cite{Arasu2013,Laur2013,Laur2011,mcsherry2009privacy,Popa2011,Xu2018}. \sysname, on the other hand, provides general-purpose, hardware-agnostic, \textit{in-situ} SQL evaluation with provable privacy guarantees and exact results.

PINQ~\cite{mcsherry2009privacy} introduced the first database with differential privacy, along with privacy budget tracking and sensitivity analysis for operator composition. We extend this work by applying its privacy calculus to {\federation}s. Follow on work in differential privacy appears in DJoin~\cite{Narayan2012}, where the system supports private execution of certain SQL operators over a federated database with strong privacy guarantees and noisy results. \sysname supports a larger set of database operators for execution and instead of using noisy results to safeguard data, \sysname uses noisy cardinalities to improve performance.

He et al.~\cite{He2017} applied computational differential privacy in join operators for private record linkage and proposed a three desiderata approach to operator execution: precision, provability, and performance. \sysname incorporates this style of join execution and approach to execution trade-offs.

Pappas et al.~\cite{Pappas2014} showed that by trading small bits of privacy for performance within provable bounds using bloom filters, systems can provide scalable DBMS queries over arbitrary boolean queries. \sysname applies this pattern of provable privacy versus performance trade-offs to the larger set of non-boolean arbitrary SQL queries.

Both Opaque~\cite{Zheng2017} and Hermetic~\cite{Xu2018} use cost models to estimate performance slowdowns due to privacy-preserving computation. In both cases, they use secure enclaves to carry out private computation, which provides constant-cost I$/$O. As such, their cost models cannot account for the variable-cost I$/$O present in \sysname.

\section{Conclusions and Future Work}
\label{sec:conclusions}

In this work, we introduced \sysname, a protocol-agnostic system for differentially-private query processing that significantly improves the performance of {\federation}s. We use a computational differential privacy mechanism to selectively leak statistical information within provable privacy bounds and reduce the size of intermediates. We introduce a novel cost model and privacy budget optimizer to maximize the privacy, performance, and accuracy trade-off. We integrate \sysname into existing {\federation} architecture and collect results using real-world medical data.

In future research, we hope to further improve the performance of {\federation}s and the efficiency of \sysname. We are currently investigating novel secure array algorithms and data structures to improve I$/$O access time, privacy budget optimizations over multiple queries, and extensions of {\federation}s using additional \mpc protocols and cryptographic primitives.  

\stitle{Acknowledgments:}
We thank Abel Kho, Satyender Goel, Katie Jackson, Jess Joseph Behrens, and the HealthLNK team for their guidance and assistance with CAPriCORN and HealthLNK data. 


\bibliographystyle{abbrv}
{\scriptsize \bibliography{smcql-dp,ref} }


\end{document}